\documentclass[11pt]{article}

\usepackage{graphics}   %---Comment out if we don't want PS support
\usepackage{endnotes}            %---Comment out if we don't want to be able to change footnotes to endnotes
\usepackage{amsfonts}            %---Comment if we don't want AMS fonts
\usepackage{amssymb}             %---Comment if we don't want AMS symbols
\usepackage{amsthm}              %---Comment if we don't want fancy theorem environments
\usepackage{amsmath}            %---Comment if we don't want various AMS math goodies
\usepackage{algorithm}           %---Comment if we don't want algorithm support
\usepackage{algorithmic}
\usepackage[letterpaper]{geometry}%---Uncomment if we want to allow paper size/margin changes
\usepackage{rotating}            %---Uncomment if we don't want to be able to rotate floats
\usepackage{multibib}
\usepackage{xcolor}

\usepackage{setspace}                %---Comment to remove space-changing stuff
\usepackage{xspace}

\newcommand{\ilogit}{\ensuremath{\mathrm{logit}^{-1}}\xspace}

\usepackage{natbib}
\bibpunct{(}{)}{;}{a}{,}{;}

\begin{document}

%---Document title page starts here
\title{A Behavioral Micro-foundation for Cross-sectional Network Models\thanks{This research was supported by US ONR award \#N00014-08-1-1015, NIH award 1R01GM144964-01, and NSF award SES-2448652.  An earlier version of this paper was presented at the 2009 meeting of the American Sociological Association.}
}

%--Author(s), date, etc.
\author{Carter T. Butts\thanks{Departments of Sociology, Statistics, Computer Science, and EECS and Institute for Mathematical Behavioral Sciences; University of California, Irvine; SSPA 2145; Irvine, CA 92697; \texttt{buttsc@uci.edu}} \and
Alexander Murray-Watters\thanks{Department of Sociology; University of California, Irvine}
}
\date{{ }} %12/7/2022}
\maketitle

%---Document abstract and keywords
\begin{abstract}
Models for cross-sectional network data have become increasingly well-developed in recent decades, and are widely used.  This has led to a growing interest in the connection between such cross-sectional models and the behavioral processes from which the corresponding networks were presumably generated.  Here, we build on prior work in this area to present a behavioral micro-foundation for cross-sectional network models, based on a continuous time stochastic choice mechanism, that can accommodate highly general classes of cases (including agents who are not themselves in the network, and multilateral edge control).  As we show, the equilibrium behavior of this process under appropriate conditions can be expressed in exponential family form, allowing estimation of individual preferences using existing methods; the graph potential separates naturally into a preference-based term reflecting agent utilities, and an entropic term reflecting the rules of tie formation.  We illustrate our approach via an analysis of friendship in a professional organization, and modeling of phase transitions in the structure of small groups.\\[5pt]
%JMS version (120 words)
%Models for cross-sectional network data have become increasingly well-developed in recent decades, with a growing interest in the connection between such models and the behavioral processes from which networks are generated.  Building on work in this area, we present a behavioral micro-foundation for cross-sectional network models based on potential games that can accommodate factors such as agents who are not themselves in the network and multilateral edge control.  The equilibrium behavior of this process can be expressed in exponential family form, allowing estimation of individual preferences using existing methods; the graph potential separates naturally into a preference-based term reflecting agent utilities, and an entropic term reflecting the rules of tie formation.  Theoretical and empirical illustrations are also shown.

\end{abstract}

%---Definitions for Defs, Theorems, etc.
\theoremstyle{plain}                        %---Comment out this line if not using amsthm
\newtheorem{axiom}{Axiom}
\newtheorem{lemma}{Lemma}
\newtheorem{theorem}{Theorem}
\newtheorem{corollary}{Corollary}
\newtheorem{proposition}{Proposition}

\theoremstyle{definition}                 %---Comment out this line if not using amsthm
\newtheorem{definition}{Definition}
\newtheorem{hypothesis}{Hypothesis}
\newtheorem{conjecture}{Conjecture}
\newtheorem{example}{Example}

\theoremstyle{remark}                    %---Comment out this line if not using amsthm
\newtheorem{remark}{Remark}

%---Actual Text Begins Here

\section{Background and Introduction}

%network formation games

Models for cross-sectional network data have become increasingly well-developed in recent decades, particularly those based on the use of discrete exponential families \citep{lusher.et.al:bk:2012,schweinberger.et.al:ss:2020}.  Such exponential family random graph models (or ERGMs) are highly effective at tasks such as identifying and quantifying structural ``biases'' (i.e., systematic tendencies towards over or underrepresentation of particular types of substructures), hypothesis testing, imputation, and generation of novel structures with similar properties to observed networks; substantively, they have also shed light on questions such as the nature of triadic closure, the connections between dependence and structural bias, and phase behavior in social systems.  One particularly important line of inquiry attempts to connect these cross-sectional models with micro-level generative processes, providing a rigorous bridge between social dynamics and the structures they produce \citep{butts:jms:2024}.   Although explicitly dynamic models provide one response to this problem (see, e.g., \citet{snijders:jms:1996,robins.pattison:jms:2001,snijders:ch:2005}), these require longitudinal data which is frequently impossible or impractical to obtain.  Alternatively, it is in some cases possible to derive cross-sectional behavior from socially plausible micro-dynamics (see \citet{butts:jms:2024} for a review), offering potential interpretations for models estimated at the cross-sectional level, and potentially enabling predictions based on comparative statics (even where dynamics cannot be directly observed).  Recent work in this area \citep{butts:jms:2019,butts:jms:2022} has for instance shown that the scaling of mean degree and/or reciprocity observed in large networks can potentially be understood in terms of a hidden mobility process (in which individuals' quotidian movements between locations and social settings shape network structure by affecting opportunities for tie formation), and likewise that explanations for these phenomena based on degree constraint (as proposed e.g. by \citet{dunbar:jhe:1992}) lead to unrealistic consequences for network structure.  

In many social contexts, it is natural to consider micro-processes based on decision making (i.e., individual choices to pursue or discontinue relationships).  Such a direction was famously pursued by Snijders (\citeyear{snijders:jms:1996,snijders:sm:2001}), who introduced a modeling approach that combined choice-based dynamics inspired by the tradition of agent-based modeling \citep{carley:asr:1991,hummon.fararo:sn:1995,skvoretz.et.al:jms:1996} with statistical ideas.  In these models (often called SAOMs, or ``stochastic actor-oriented models''), ties are viewed as reflecting the choices of individuals within a network, who unilaterally make decisions to either send or remove ties to alters via a myopic stochastic choice process.  Although the bulk of this line of work deals with models for longitudinal data, \citet{snijders:sm:2001} also shows that, in some cases, it is possible to derive a cross-sectional ERGM form from a SAOM process - in particular, that if such a process is observed at a random time, the network will have an ERGM distribution that depends on actor utilities in a deterministic way.  Mele \citep{mele2017structural,mele2022structural,gaonkar2021model} further develops this theme, developing a framework where a subclass of ERGMs can be interpreted as a measurement of an agent-based process at equilibrium, where an edge connects two agents only when they mutually agree that it exists.  These results suggest considerable potential for further insights from generative choice models for group structure.

Progress to date on connecting choice processes with cross-sectional structure opens the possibility of further generalization.  Most current models assume unilateral edge control: i.e., ties are the direct result of the choice of a single individual, who can unilaterally determine whether a relationship is present or not.  For some types of relationships, this is arguably a reasonable approximation; e.g., a worker in an organization can unilaterally decide to seek advice from a potential alter, whether or not the alter consents.  In other settings, this approximation is more dubious.  For instance, mutual interaction in a cocktail party or other open setting can be initiated by either party (whether or not the other wishes to interact), friendship ties in the sense of \emph{philos} \citep{krackhardt:ch:1992} must be shared by both participants, and ties on many social media platforms require the assent of both parties.  Multilateral edge control is thus a feature of many real relations, motivating the development of schemes for representing networks in which relationships arise from multiple distinct decisions, with relational norms that govern how those decisions combine to determine relational status.  A second feature of current models is the identification of the decision makers who drive network structure with the nodes of the network itself: it is assumed that the network is comprised of agents, who make decisions regarding their own relationships with each other.  While this is an important special case, other types of settings also exist.  For instance, in an organization, managers or supervisors may be able to mandate (or possibly prohibit) certain kinds of interactions or relationships among their subordinates.  More subtly, nodes within the network may not be agents at all: children at play may argue over the fictive relationships among inanimate objects such as dolls, pundits and activists may spar over the ways in which political concepts should be connected within a discourse domain, and planning committees may haggle over ties within infrastructure networks.  Encompassing such scenarios requires the ability to separate the notion of the decision-making agent from the nodes within the network of interest, with overlap between the two sets being possible but non-essential.  Distinct agent/node sets further underscore the need to support complex patterns of multilateral edge control, since numerous parties could be involved in e.g. determining edge states in a negotiated network, with different parties weighing in on different sets of edges.

In the remainder of this paper, we provide a choice-based micro-dynamic foundation for certain classes of exponential family random graph models, that incorporates both multilateral edge control and agent/node distinctions.  Our approach builds on and generalizes prior work in the area, taking existing models as special cases.  In addition to providing a road to statistical inference for behavioral tendencies when dynamic information is unavailable, our framework also provides new theoretical insights into the ways that preferences and social norms can combine to shape network structure, and in some cases support comparative statics predictions for how structures would be expected to differ if the norms of tie formation were changed.  These results further contribute to recent work on the statistical mechanics of networks \citep[e.g,][]{robins.et.al:ajs:2005,butts:jms:2021,butts:jms:2024,diessner.et.al:jctc:2024} that separates the role of entropic effects and the social forces that influence network structure, demonstrating that tie formation norms lead to entropy terms while preferences behave analogously to energy terms in physical models.  We illustrate the framework by application to an empirical case of friendship ties among members of a law firm, and to more abstract models of phase changes in the structure of small groups.  

Before describing our choice-theoretic framework, we begin by establishing some of the core formalisms and notation required, and by reviewing relevant basics of exponential family random graph models.  We describe our specific development in Section~\ref{sec_choice}, followed by an illustrative empirical application in Section~\ref{sec_application} and an application to group structure in Section~\ref{sec_phase}.  Section~\ref{sec_conclusion} concludes the paper.

\subsection{Formal Preliminaries}

The networks that concern us here can be represented as simple directed or undirected graphs of fixed order.  We represent both by ordered pairs $G=(V,E)$, where $V$ is the set of \emph{vertices} and $E$ is a set of \emph{edges} on $V$.  Throughout, we shall use $n=|V|$ to refer to the order of the graph being described (where $|\cdot|$ denotes cardinality).  For simple graphs representing undirected relations, the elements of $E$ are two-element subsets of $V$; in the directed case, the elements of $E$ are ordered pairs of vertices.  Frequently, we will represent graphs via their adjacency matrices.  The \emph{adjacency matrix} of a graph, $G$, is the $n \times n$ binary matrix, $y$, such that $y_{ij}=1$ if $\{i,j\} \in E$ in the undirected case (or $(i,j) \in E$ in the directed case) and $y_{ij}=0$ otherwise.  Given a set of multiple graphs $G_1=(V,E_1),\ldots,G_m=(V,E_m)$ on common vertex set $V$, we can easily extend this notation by defining the corresponding \emph{adjacency array}, $y$, as the $m \times n \times n$ binary array such that $y_{ijk}=1$ if $\{j,k\} \in E_i$ ($(j,k) \in E_i$ in the directed case) and $y_{ijk}=0$ otherwise.  The adjacency matrix for the $i$th graph in such an array clearly corresponds to $y_{i\cdot\cdot}$, and is referred to as the $i$th \emph{slice} of $y$.

Adjacency matrices and arrays are particularly useful for describing \emph{random graphs}, i.e., random variables whose state space consists of a graph set.  We will here designate random graphs via their adjacency matrices, using capital letters to denote random variables and lower case letters to denote realizations.  Thus, a random graph $G$ may be represented by (random) adjacency matrix $Y$, within which the presence of an $i,j$ edge is denoted by the binary random variable (or \emph{edge variable}) $Y_{ij}$.  A corresponding realization of $G$ would then be represented by the adjacency matrix $y$, with the state of the $i,j$ edge indexed by $y_{ij}$.  This notation extends directly to adjacency arrays.  In some cases, we will also want to refer to edge variables within a random graph \emph{excluding} a particular edge variable, or with the states of particular variables held constant.  To this end, we define $Y^c_{ij}$ to be the set of all edge variables of $Y$ \emph{other than} $\{i,j\}$ (directed case, $(i,j)$); likewise, $y^c_{ij}$ refers to all elements of realization $y$ other than the $i,j$th.  A related situation arises when we wish to describe a graph structure in which a particular edge is forced to be either present or absent: $Y^+_{ij}$ and $Y^-_{ij}$ are respectively employed to denote random matrices with $Y^+_{ij}=1$ and $Y^-_{ij}=0$ and $Y^+_{kl}=Y^-_{kl}=Y_{kl}$ for $\{i,j\} \neq \{k,l\}$ (directed case, $(i,j) \neq (k,l)$).  Realization matrices $y^+_{ij}$ and $y^-_{ij}$ are defined in like manner.  

Finally, we note that when describing stochastic processes on graph sets (i.e., random graph processes), we will employ sequences of random adjacency matrices or arrays indexed by parenthetical superscripts (e.g., $Y^{(1)},Y^{(2)},\ldots$).  The same principle will be applied to realizations of such processes; hence, a realization of the above process might be denoted $y^{(1)},y^{(2)},\ldots$.  For processes on graph sets, we simply apply the same notation to the corresponding adjacency arrays.

\section{Cross-sectional Network Models}

Cross-sectional network models have been the subject of considerable research over the past several decades (see, e.g., \citet{holland.leinhardt:jasa:1981,frank.strauss:jasa:1986,wasserman.pattison:p:1996,pattison.robins:sm:2002,wasserman.robins:ch:2005,snijders.et.al:sm:2006,hunter.handcock:jcgs:2006,hunter.et.al:jcgs:2012,lusher.et.al:bk:2012,schweinberger.et.al:ss:2020} and related references), one result of which has been the emergence of a general family of formalisms for representing probability distributions on graphs.  Specifically, let $Y$ be the adjacency matrix of order-$n$ random graph $G$, and let $\mathcal{Y}_n$ be the support of $Y$.  We may then write the pmf of $Y$ in exponential family form as 
\begin{equation}
\Pr\left(Y=y\left|\rho\right.\right) = \frac{\exp\left(\rho\left(y\right)\right)}{\sum_{y' \in \mathcal{Y}_n} \exp\left(\rho\left(y'\right)\right)} I_{\mathcal{Y}_n}(y), \label{eq_erg}
\end{equation}
where $\rho: \mathcal{Y}_n \mapsto \mathbb{R}$ is a \emph{graph potential}, and $I_{\mathcal{Y}_n}$ is an indicator function for membership in $\mathcal{Y}_n$.  Typically, the graph potential is expressed in canonical form as $\rho\left(y\right)=\theta^T t\left(y\right) + \log h(y)$, with $t: \mathcal{Y}_n \mapsto \mathbb{R}^p$ being a vector of \emph{sufficient statistics}, $\theta \in \mathbb{R}^p$ a vector of \emph{parameters}, and $h:\mathcal{Y} \mapsto \mathbb{R}_{\ge 0}$ a \emph{reference measure} that controls the baseline behavior of the distribution.  It is common when working with single, unvalued social networks with fixed $V$ to take $h(y) \propto 1$, i.e. the counting measure on $\mathcal{Y}$, in which case $Y$ approaches the uniform distribution on $\mathcal{Y}$ as $\theta \to 0$ (where 0 should be understood to be the 0-vector).  However, the reference measure can in some cases play an important role in quantifying differences in opportunities for tie formation in the underlying generative process, which may motivate alternatives to the counting measure on theoretical grounds \citep{butts:jms:2019,butts:jms:2022}.  Below, we show that such non-trivial reference measures arise naturally in the context of multi-party edge control, with implications for the impact of alternative relational norms on social structure.

Although sometimes loosely described as a model in its own right, the framework of Equation~\ref{eq_erg} is more properly regarded as a general way of representing distributions on $\mathcal{Y}_n$; in particular, observe that by choosing $t$ to correspond to a $|\mathcal{Y}_n|$-dimensional vector of indicators for the state of $Y$, any pmf on a graph set of fixed order may be written in this fashion.  Notwithstanding the potential for confusion, it is both traditional and useful to refer to network models parameterized in exponential family form as \emph{exponential random graph models} (or ERGMs), and we follow this practice here.

While the above shows the single-network case, extensions to the multiple-network case are immediate.  Indeed, simply taking $Y$ and $\mathcal{Y}_n$ to refer to $m$-slice adjacency arrays allows Equation~\ref{eq_erg} to be applied to graph sets without additional difficulties.  General issues relating to parameterization of such ``multivariate'' exponential random graph (or MERG) models were extensively elucidated by \citet{wasserman.pattison:p:1996,pattison.wasserman:bjmsp:1999} and \citet{robins.et.al:p:1999}, although their use dates back to the original work by \citet{holland.leinhardt:jasa:1981}.  

The attractiveness of ERGMs as a general framework for cross-sectional network modeling stems from several sources.  As already noted, the ERGM formalism is \emph{complete} on the set of finite order graphs and graph sets, in the sense that all distributions for such entities can be written (albeit not always parsimoniously) in ERGM form.  As special cases of the more general discrete exponential families \citep{brown:bk:1986,barndorffnielsen:bk:1978}, considerable inferential theory exists for ERGMs (see, e.g., \citet{hunter.et.al:jcgs:2012,lusher.et.al:bk:2012,schweinberger.et.al:ss:2020} for reviews); thus, fairly general methods of parameter estimation, hypothesis testing, model selection, and adequacy evaluation exist for models parameterized in this form.  Similarly, general algorithms for simulation of ERGMs exist (e.g., \citet{snijders:joss:2002,hunter.et.al:jss:2008}) based on Markov chain Monte Carlo (MCMC) methods, as well as more specialized exact \citep{butts:jms:2018} and approximate \citep{butts:jms:2015} samplers.  Finally, the ERGM form is well-suited to parameterization based on \emph{dependence hypotheses}, a notion from spatial statistics \citep{besag:jrssB:1974} which can be adapted to produce families of network models satisfying particular assumptions regarding the interdependence of edges (see, e.g., \citet{frank.strauss:jasa:1986,pattison.robins:sm:2002}).  

As these observations suggest, the ERGM framework acts as a sort of ``lingua franca'' for cross-sectional network modeling: models translated into ERGM form can immediately exploit a broad family of inferential and simulation tools, and can moreover be compared and even combined on a principled basis.  While the creation and elaboration of this framework has represented a tremendous advance, it cannot immediately address the \emph{content} of network models -- this last must necessarily come from a combination of empirical investigation, prior substantive knowledge regarding the structures under study, and an enumeration of modeling objectives.  Ideally, one would like to be able to derive ERGM forms from first principles, based on underlying substantive theory regarding the generative processes governing the network under study.  Encouraging progress has been made in this direction.  Apart from the dependence processes and work on reference measures cited above, there has been some success in linking ERGMs to both choice-theoretic \citep{snijders:sm:2001,mele2017structural} and physical \citep{grazioli.et.al:jpcB:2019,diessner.et.al:jpcB:2023,diessner.et.al:jctc:2024} processes.  The major portion of the text that follows is concerned with this problem.  Before turning to this, however, we review several other formal ideas required for our development.  Readers already familiar with ERGMs and related Markov chain Monte Carlo methods may wish to proceed to Sec.~\ref{sec_choice}, where we discuss our choice-theoretic formulation.

\subsection{Potentials, Conditional Probabilities, and ERGMs}

Equation~\ref{eq_erg} specifies the joint probability of an entire graph (or set thereof) as a single observation.  Direct calculation of this quantity is generally infeasible, due to the normalizing factor $Z\left(\rho,\mathcal{Y}_n\right)=\sum_{y' \in \mathcal{Y}_n} \exp\left(\rho\left(y'\right)\right)$.  $Z$ corresponds directly to the \emph{partition function} of statistical mechanics, and (through its derivatives) implicitly contains a wealth of information regarding the behavior of the corresponding model.  (See \citet{strauss:siam:1986,butts:ch:2007} for related discussion.)  Regrettably, $Z$ has closed-form expressions for only a handful of models, and is generally incomputable due to the large number of elements in the corresponding sum (and the roughness of the underlying function).  Various computational strategies are employed to deal with this issue \citep{hunter.et.al:jcgs:2012,schweinberger.et.al:ss:2020,krivitsky.et.al:jss:2023}.  However, the problem of normalization does motivate the general interest in conditional probabilities and probability ratios vis a vis models in ERGM form.  For instance, let $y,y'\in \mathcal{Y}_n$ be elements of the support of $Y$.  Then, from Equation~\ref{eq_erg},
\begin{align}
\frac{\Pr\left(Y=y'\left|\rho \right.\right)}{\Pr\left(Y=y\left|\rho \right.\right)} &= \frac{\exp\left(\rho\left(y'\right)\right)}{\sum_{y'' \in \mathcal{Y}_n} \exp\left(\rho\left(y''\right)\right)} \frac{\sum_{y'' \in \mathcal{Y}_n} \exp\left(\rho\left(y''\right)\right)}{\exp\left(\rho\left(y\right)\right)}\\
&=\exp\left(\rho\left(y\right)-\rho\left(y\right)\right). \label{e_ergrat}
\end{align}
Thus, the log-odds of $y'$ versus $y$ reduces to their difference in potentials, with $Z$ falling out of the equation.  This observation leads immediately to an expression for the conditional probability of a single edge, given the rest of the graph:
\begin{align}
\frac{\Pr\left(Y=y^+_{ij}\left|y^c_{ij},\rho \right.\right)}{\Pr\left(Y=y^-_{ij}\left|y^c_{ij},\rho \right.\right)} &= \exp\left(\rho\left(y^+_{ij}\right)-\rho\left(y^-_{ij}\right)\right)\\
\frac{\Pr\left(Y=y^+_{ij}\left|y^c_{ij},\rho \right.\right)}{1-\Pr\left(Y=y^+_{ij}\left|y^c_{ij},\rho \right.\right)} &= \exp\left(\rho\left(y^+_{ij}\right)-\rho\left(y^-_{ij}\right)\right)\\
\Pr\left(Y=y^+_{ij}\left|y^c_{ij},\rho \right.\right) &= \frac{\exp\left(\rho\left(y^+_{ij}\right)-\rho\left(y^-_{ij}\right)\right)}{1+\exp\left(\rho\left(y^+_{ij}\right)-\rho\left(y^-_{ij}\right)\right)}\\
&=\frac{1}{1+\exp\left(\rho\left(y^-_{ij}\right)-\rho\left(y^+_{ij}\right)\right)}\\
&=\ilogit\left(\rho\left(y^+_{ij}\right)-\rho\left(y^-_{ij}\right)\right). \label{e_ergcond}
\end{align}
Here again, we are left with a simple expression in terms of potential differences -- in this case, the conditional probability of an edge is seen to be the inverse logit of the change in potential associated with that edge (holding the rest of the graph constant).  This is, of course, the basis for the well-known conditional logistic interpretation of the ERGM, and for the standard ERGM pseudo-likelihood estimator \citep{frank.strauss:jasa:1986,strauss.ikeda:jasa:1990}.  Our use for the expression is rather different: as we shall see, it forms an important element in our linkage of behavioral choice models with network structure.  Before coming to this, however, we round out this section of the paper by reviewing the connection between cross-sectional ERGMs and a basic family of generative processes.

\subsection{Generative Processes}

Although we have thus far discussed ERGMs entirely within a cross-sectional framework, they are also associated with various generative processes (see \citet{butts:jms:2024} for a review).  Most generically, it is obviously the case that a temporal aggregation or cross-section of any fixed-order network process can described in the form of Equation~\ref{eq_erg}; this is a trivial consequence of the already-noted fact that all distributions on $\mathcal{Y}_n$ can be written in this form.  Such an observation is not especially helpful in linking ERGMs with generative processes, however, since it provides no indication of how the appropriate potential, $\rho$, might be found.  Moreover, graph potentials for aggregations or cross-sections of dynamic processes will not in general be invariant to sampling time, length of aggregation window, or even the initial conditions of the underlying process \citep{lerner.et.al:jmp:2013}.  Thus, one cannot hope to find a completely generic way to map generative processes onto ERGMs (much less the reverse).

Despite this caveat, there are certain families of generative processes that do give rise to cross-sectional distributions in a predictable manner.  Typically, these processes have been studied in the context of statistical simulation of ERGM draws, e.g. in support of inference or model assessment (e.g. \citet{crouch.et.al:pres:1998,snijders:joss:2002,hunter.et.al:jss:2008}), with a focus on properties relating directly to such goals (e.g. simulation accuracy and efficiency).  Our interest is somewhat different: we will ultimately use the generative process to link micro-level behavior with cross-sectional network structure.  Nevertheless, it happens that we can exploit this pre-existing work in furthering our theoretical objectives.  To foreshadow, these and closely related developments have been employed by e.g. \citet{snijders:sm:2001}, \citet{mele2017structural}, and \citet{koskinen.lomi:jsp:2013} to obtain cross-sectional distributions from micro-level behavioral processes, an approach we partially recapitulate here.

We start by invoking a particularly simple and well-known process, namely the single-update Gibbs sampler on graphs \citep{snijders:joss:2002}.  Let us begin by positing some support for a process with cross-section $Y$, $\mathcal{Y}_n$, from which some initial graph $Y^{(0)}$ is drawn from an ERG with arbitrary potential $\rho_0$.  Let $E^*(\mathcal{Y}_n)$ denote the set of edge variables of $Y$.  We now form the sequence $Y^{(1)},Y^{(2)},\ldots$ via the following iterative procedure:
\begin{enumerate}
\item At the $i$th iteration, draw an edge variable $Y_{jk}$ from a variable selection process $S_i$ on $E^*(\mathcal{Y}_n)$;
\item With probability $Pr\left(Y^{(i-1)=}\left(y^+_{jk}\right)^{(i-1)}\left|\left(Y^c_{jk}\right)^{(i-1)}=\left(y^c_{jk}\right)^{(i-1)},\rho\right.\right)$, let $Y^{(i)}=\left(Y^+_{jk}\right)^{(i-1)}$.  Otherwise, let $Y^{(i)}=\left(Y^-_{jk}\right)^{(i-1)}$.
\end{enumerate}
The following is then a standard result (e.g., \citet{gilks.et.al:ch:1996}):
\begin{proposition} \label{prop_gibbs}
Let $\rho$ be a finite potential on $\mathcal{Y}_n$, and let $S_i$ be a process on $E^*(\mathcal{Y}_n)$ such that
\begin{enumerate}
\item $S_i$ is independent of $Y^{(1)},Y^{(2)},\ldots,Y^{(i)}$; and
\item $\sum_{l=1}^i I(S_l=Y_{jk}) \to \infty$ as $i \to \infty$ almost surely, for all $j,k$ such that $Y_{jk} \in E^*(\mathcal{Y}_n)$.
\end{enumerate}
Then the sequence $Y^{(0)},Y^{(2)},\ldots,Y^{(i)}$ approaches ERG$(\rho)$ in distribution as $i\to \infty$.
\end{proposition}
As this is a special case of the well-known random scan Gibbs sampler, we do not provide a proof of this proposition here.  Interested parties are referred to the above review chapter and included references for details.

Proposition~\ref{prop_gibbs} holds because the iterative procedure forms a Markov chain on the elements of $\mathcal{Y}_n$ whose equilibrium distribution is ERG$(\rho)$.  Irreducibility and aperiodicity of the chain stems from the fact that the conditional probability of setting an edge to be present or absent is bounded away from 0 or 1 (itself a consequence of the finiteness of $\rho$ combined with Equation~\ref{e_ergcond}), along with the properties of $S$; that the resulting equilibrium distribution is the correct one is less obvious, but is a known result for chains of this type.

For our purposes, several points bear emphasis.  First, the above procedure provides an example of a generative process based on local structural adjustments, which ultimately gives rise to a known cross-sectional distribution.  Second, the local adjustments on which the process is based are Bernoulli trials, whose probabilities are inversely logistic in the potential difference between successive graphs.   Third, the selection of the initial graph state is wholly arbitrary, in the sense that it does not impact the equilibrium distribution of the process.  Finally, the asymptotic cross-sectional properties of the graph process do not depend upon the details of the process by which edge variables are selected for updating -- so long as the specified conditions are met, edge variables may be updated in any order (and, indeed, irregularly) without affecting the long-run behavior of the chain.  Taken together, these properties suggest an interesting possibility: if it were possible to map a plausible behavioral process to the sampling algorithm, then one could in principle use its equilibrium distribution to predict the long-run cross-sectional behavior of the corresponding social system.  This has turned out to be a fruitful notion in work on network microprocesses \citep{butts:jms:2024}, which we exploit in the following section.

\section{Cross-sectional Distributions from Stochastic Choice} \label{sec_choice}

In the previous section of this paper we focused on the definition of cross-sectional network models, closing with the description of a simple family of stochastic processes whose equilibria correspond to specified ERGM distributions.  In this section, we build on these results by demonstrating a family of behaviorally reasonable micro-processes that belong to the above class, and specify sufficient conditions for relating their resulting equilibria to individual preferences.  We begin by describing a simple stochastic choice process for social networks, subsequently showing how this process can lead to a specified equilibrium distribution in the case of unilateral edge control (i.e., relations for which any given edge is controlled by a single party).  Next, we generalize this to the more complex case of bilateral edge control (i.e., relations for which any given edge is jointly controlled by two parties), closing the section with a framework for the general multilateral case.  These results build on past work on networks with unilateral edge control \citep{snijders:sm:2001,koskinen.lomi:jsp:2013} and bilaterial edge control with mutual assent \citep{mele2017structural,mele2022structural,gaonkar2021model}, allowing for general edge control regimes in which the agents who govern the network may or may not be synonymous with its vertices.  Relatedly, we also introduce a system of context-agnostic nomenclature to describe these various cases without tacitly assuming properties of the agents or of relational content.

\subsection{Agent Choice Framework}

The core of our behavioral micro-process is a stochastic choice model, through which individuals are presumed to make decisions that impact (but may or may not entirely determine) a set of relationships.  We begin by assuming a \emph{manifest relation}, $Y$, on some fixed vertex set $V$.  The edges of $Y$ are presumed to arise from the actions of one or more \emph{agents,} who are collectively denoted by the set $A$.  It is important to note that the set of agents need not be the same as (nor even have a non-empty intersection with) the set of vertices, although $V=A$ is a useful special case.  The connection between agents and the edges they affect is formalized via a set of \emph{control lists}, which specify the agents whose behaviors govern any given edge.  Specifically, for each edge variable $Y_{ij}$, we define the corresponding control list $c_{ij}\subseteq A$ to be the minimum lexically ordered $\ell$-tuple of agents whose behaviors determine the state of $Y_{ij}$.  (The assumption of lexical ordering is not technically necessary, but considerably simplifies subsequent notation.)  The control list length, $\ell$, is referred to as the \emph{control number} for $Y$, and is of considerable substantive importance.  Networks for which $\ell=1$ are said to exhibit \emph{unilateral edge control}, as the state of each edge is determined by the behavior of a single agent.  Similarly, networks with $\ell=2$ are said to exhibit \emph{bilateral edge control}, with $\ell>2$ reflecting the more general multilateral case.  Particular issues arise in each of these settings, and we therefore treat them in greater detail below.  In passing, we note that one can in principle posit cases of \emph{mixed control}, in which $\ell$ varies across edge variables.  Although we here limit ourselves to the ``pure'' cases in which $\ell$ is constant, our framework generalizes to mixed control in a fairly natural way.

Given an agent set and associated control structure, we must next posit a framework for representing agent choices regarding his or her associated edges.  Here, we propose to do this via a set of $\ell$ latent graphs, each of which encodes the current decision of a set of agents regarding the edges of $Y$.  Specifically, we define the \emph{prosphoric array} for $Y$ to be an $\ell \times N \times N$ adjacency array, $P$, whose $i,j,k$ cell contains the current decision state of the $i$th agent in $c_{jk}$ regarding $Y_{jk}$.  The manner in which the cells $P_{\cdot jk}$ determine $Y_{jk}$ -- and hence the behavioral interpretation of the associated ``decision'' -- is relation-specific, and will be discussed in detail below.  Intuitively, however, it is often useful to think of the prosphoric array as representing ``offers'' or ``overtures'' for the creation or maintenance of ties, whose appropriate coincidence (depending on both $\ell$ and on the properties of the relation) may lead to the presence of the associated manifest relation.  While this analogy is the basis for the terminology employed here (``prosphoric'' from the Greek for ``offer'' or ``overture''), it should be stressed that we do not presume that such sub-relational interactions are actually taking place; rather, we consider the prosphoric array as a formal device for aggregating agent choices, which may or may not have any direct relationship to the physical or social details of the interactive process through which the ties of $Y$ are actually created.  In particular, it should be noted that the slice structure of $P$ is determined by the lexical ordering of $A$, and is thus semi-arbitrary.  That said, the joint structure of $P$ \emph{is} meaningful in terms of its relationship to the manifest relation, $Y$, and it is to this that we now turn.

\subsubsection{Unilateral Edge Control}

To express the relationship between individual decisions and the resultant graph structure, we posit a \emph{resolution function}, $r$, which takes the set of possible prosphoric arrays into $\mathcal{Y}_n$ (the set of possible manifest networks).  In the simplest case of unilateral edge control (i.e., $\ell=1$), $r$ is simply the identity function: that is, $Y_{jk}=P_{1jk}$.  As this implies, a manifest edge under unilateral control exists if and only if its controlling agent chooses it, and hence the meaning of a decision within the prosphoric array is simply that of whether or not to maintain the corresponding relationship.

Unilateral edge control is an obvious model for directed relations such as advice seeking, in which each individual has direct control over his or her outgoing ties.  In this case (studied e.g. by \citet{snijders:sm:2001}) it is reasonable to let $V=A$, and to set $c_{ij}=(j)$ for all $(i,j)$ pairs.  While such examples are easy to come by, they are not the only phenomena which can be modeled in this way.  For instance, consider an organizational design problem, in which managers within different divisions of an organization have discretion to define the formal reporting relationships for their own divisions (with the CEO defining reporting relationships among the division heads).  In this case, $A \subset V$, and $c_{ij}$ for a given relationship would be the identity of the controlling manager.  One can easily conceive of related applications on undirected relations (e.g., design of communication channels rather than reporting relations), relations for which the set of agents is not contained within the set of vertices (e.g., design of physical systems, such as road networks), etc.  What all of these cases have in common is that edges are contingent upon the behavior of a single decision maker, and that the corresponding decision is simply one of whether or not to maintain the edge in question.

\subsubsection{Bilateral Edge Control}

Although many relations are unilateral, there are likewise many examples of relations that depend upon the actions of two parties.  In such cases, $\ell=2$, and more than one substantively distinct mapping from $P$ to $Y$ is possible.  Of the possible variants, two forms are of clear sociological interest and are discussed in detail here.  First, we have the case in which a manifest edge exists if and only if both controlling parties agree to it; we refer to these as \emph{symphonic relations} (from ``symfonia,'' or agreement), and define the corresponding resolution function as $r(P)_{jk}=P_{1jk}P_{2jk}$.  The second case is that in which a manifest edge exists if and only if \emph{some} controlling party agrees to it.  We refer to such relations as \emph{epibolic} (from ``epiboli,'' or imposition), expressing the intuition that either party can ``impose'' the relationship upon the other.  The resolution function for an epibolic relation is $r(P)_{jk}=1-\left(1-P_{1jk}\right)\left(1-P_{2jk}\right)$, i.e., a logical OR on the elements of $P_{\cdot jk}$.  As these examples make clear, we can usefully view the resolution function as describing the \emph{relational norms} that determine how relationships are defined within the setting of interest.  Although one can identify other types of bilateral relational norms beyond these two (e.g., a logical XOR), symphonic and epibolic relations would seem to constitute the bulk of those usually encountered in sociological settings.  Given this, we take a moment to consider each in turn before proceeding to the case of general multilateral relations.

\paragraph{Symphonic Relations}

As with unilateral relations, bilateral symphonic relations are obvious models for a wide array of interpersonal networks.  For instance, many notions of friendship (e.g., the \emph{philos} of \citet{krackhardt:ch:1992}) are defined in such a way as to require mutual assent, and mutual agreement is of practical necessity for many forms of collaboration (e.g., in task performance following disasters% \citep{butts.et.al:joss:2009}
).  The typical scenario for such relations is an undirected graph, $Y$, with $A=V$ and $c_{jk}=(j,k)$ (where $j<k$ in lexical order).  Although the above gives the impression that symphonic relations occur only in cases of positive interaction, this is not the case.  Consider, for instance, the interaction between an infrastructure designer or maintainer (e.g., a civil authority) and an adversary (e.g., a hostile military force) in the context of a transportation network.  For an edge (e.g., a roadway) to be passable, it must both be built/maintained by the infrastructure designer \emph{and} allowed to operate by the adversary.  Obviously, the designer and adversary have very different goals, and their interaction is strongly negative.  Nevertheless, such a situation can be treated as a symphonic relation with $A$ consisting only of the designer and adversary, and $c_{jk}$ being the ordered elements of $A$ for all $j,k$.  

As with the unilateral case, there is no necessary correspondence between the resolution mechanism and the directedness of the edges of $Y$, the positivity/negativity of the relation, nor the elements of $A$.  Modeling $Y$ as symphonic does, however, fix the meaning of the decision $P_{ijk}=1$ to be ``the $i$th controlling party of the $j,k$ edge chooses to implement the edge, \emph{if} the other party concurs.''  So long as one bears in mind that ``concurrence'' may be tacit (and need not connote a positive-valence interaction among the controlling parties), this provides a clear sense of the decision in question.

\paragraph{Epibolic Relations}

If friendship is an obvious example of a symphonic relation, interpersonal contact is a similarly common epibolic example.  Contact is inherently symmetric, and may be instigated by either party to the relation; it is thus natural to model it by letting $A=V$, and $c_{jk}=(j,k)$.  Interestingly, it may be appropriate to model the same underlying relationship as epibolic or unilateral, depending on what is assumed to drive the behaviors of the participants.  For instance, an advising relation generally involves a two-way exchange of information, even if one party is the initiator of the interaction.  If this exchange \emph{per se} were conjectured to drive agent behavior (and not the properties of the unilateral advice-seeking edge), then it would be more appropriate to treat the interaction as undirected and epibolic rather than directed and unilateral.  (Use of model comparison techniques to empirically disambiguate such cases is in principle possible, although this will not be treated here.)  While contact and advice are directly interpersonal, one can naturally find epibolic relations in other contexts.  Returning to our earlier transportation example, we could for instance consider a scenario in which two authorities (e.g., governments of neighboring city-states) have the capacity to build road segments within the same area (e.g., neutrally held territory adjacent to both).  Since a road segment exists if implemented by either party (in the absence of malfeasance), the road network as a whole is an epibolic relation.  The agents in this case are the two government authorities, and the control list for all edges is simply these two agents in lexical order.

The meaning of the decision $P_{ijk}=1$ in the epibolic case is straightforwardly implied by the above examples: ``the $i$th controlling party of the $j,k$ edge chooses to implement the edge, \emph{regardless} of whether the other party concurs.''  This is obviously close to the unilateral case, save in that neither party can prevent the edge from being established by the other.  As before, it should be stressed that this is independent of the valance of $Y$, and could reflect either tacit or explicit sub-relational interaction.

\subsubsection{General Multilateral Edge Control}

Although most studied relations fall into the unilateral or bilateral cases, there is in principle no limit to the number of agents who may be involved in edge creation.  For instance, \citet{butts.carley:jms:2007} discuss games in which a committee of designers attempts to negotiate the structure of an organization; although their case is concerned with dynamics (i.e., plans for changing an organization over time), similar interactions may occur in negotiations over structure per se.  Models for hypergraphs (which can be represented two-mode bipartite graphs, with appropriate restrictions on $\mathcal{Y}_n$) may also suggest a more general multilateral control structure.

In such cases, the number of possible resolution mechanisms is obviously quite large, and we do not attempt to enumerate them here.  Symphonic and epibolic relations generalize naturally to the general context, with respective $R$ functions $r(P)_{jk}=\prod_{i=1}^\ell P_{ijk}$ and $r(P)_{jk}=1-\prod_{i=1}^\ell \left(1-P_{ijk}\right)$.  Other potential rules include thresholding mechanisms (i.e., some critical number $k$ of controlling parties must choose an edge to activate it), which takes the ``democratic'' process of majority rule as a special case.  As before, choice of $r$ is a substantive modeling decision which indirectly determines the behavioral meaning of the prosphoric array, and as such should be guided by prior knowledge regarding the social system in question.

\subsection{Edge Updating Events, and the Decision Model}

Having constructed a framework for representing agent decisions regarding edges within a network, and having discussed the mechanisms by which those decisions are transformed into manifest ties, we are now in a position to discuss the choice process itself.  As with prior work on stochastic choice models for social networks, our approach is to model agents as boundedly rational decision makers, who myopically adjust their relationships so as to (on average) increase their utilities for the resulting manifest network.   Adjustments are assumed to occur episodically in continuous time, such that adjustment opportunities (though not realized adjustments) are independent of the current network state.  As we will subsequently show, this process (for certain families of utility functions) leads to cross-sectional behavior which can be directly parameterized in terms of individual utility functions.

\subsubsection{Edge Updating Process}

As noted above, we assume that relationships within $Y$ are persistent (as opposed to instantaneous; see \citet{butts:sm:2008} for a discussion), with changes occurring in continuous time.  Changes occur when provoked by agents' decisions (within $P$), as determined by the resolution function, $r$, encoding relational norms.  As cognitively bounded individuals, we presume that agents are not capable of evaluating all controlled edges simultaneously, nor at all times.  Rather, we presume that an unobserved saliency-governing process episodically leads each agent to evaluate his or her decision regarding a particular edge, at which time he or she revisits the decision within the current structural context (as described below).  Should this evaluation lead to a state change, $P$ (and, if necessary, $Y$) is immediately updated, and all decisions are then fixed until the next opportunity arises.  

We here refer to such ``opportunities'' as \emph{edge updating events}, and the process which generates them as an \emph{edge updating process}.  Formally, we define this process, $X^{(1)},X^{(2)},\ldots$, as a sequence of random $(j,k,l,t)$ tuples, where $a(X^{(i)})=j$ is the updating agent, $e_s(X^{(i)})=k$ is the sender of the hypothetical edge being updated, $e_r(X^{(i)})=l$ is the receiver of the hypothetical edge, and $\tau(X^{(i)})=t$ is the time at which the updating event occurs.  Like $P$, we consider $X$ to be fully latent, and we specify only that it satisfy the following two conditions.  First, we require $X$ to be independent of $P$ (and hence of $Y$) -- change opportunities (but not realized changes) do not depend on the state of the network.  Second, we require that all elements of $P$ are updated at least occasionally, in the limited sense that $\sum_{x: \tau(x)<t} I\left(a\left(x\right)=i,e_s(x)=j,e_r(x)=k\right) \to \infty$ as $t\to\infty$ a.s. for all $\{j,k\}$ (directed case $(j,k)$) in $E^*(\mathcal{Y}_n)$ and all $i \in c_{jk}$.  Thus, in the limit, all edges will be updated infinitely many times, though this process need not happen in any particular order.  While very weak, we will see that these conditions are sufficient to yield the desired result. 

\subsubsection{Stochastic Choice Mechanism}

When an edge updating event $X^{(i)}=(j,k,l,t)$ occurs, the evaluating agent $a=c_{klj}$ revisits his or her decision regarding the state of the $j,k$ edge (i.e., the value of $P_{jkl}$).  Let $P^{(i-1)}$ be the current state of the prosphoric array; $a$ must then choose between $P^{(i)}=\left(P^{(i-1)}\right)^+_{jkl}$ and $P^{(i)}=\left(P^{(i-1)}\right)^-_{jkl}$.  We here assume that $a$'s preferences can be expressed by a finite utility function on the states of $Y$, i.e., $u_a: \mathcal{Y}_n \mapsto \mathbb{R}$.  Given this, $a$ myopically chooses the state of $P_{jkl}$ based on its immediate relational consequences.  Within such a setting, a frequently used behavioral choice model is the logistic choice model \citep{mcfadden:ch:1973}, so-called because of its functional form.  Specifically, the logistic choice model predicts that the probability of choosing a given option (versus a single alternative) is equal to the inverse logit of the difference in utilities between the alternatives.  When one option is decisively superior, the model converges towards the rational actor solution (strict utility maximization), alternately converging towards equal probability of selection at the indifference point; between these two extremes, the decision maker is predicted to favor the superior option, with the degree of favoring increasing with the difference in utilities.  Applying this framework to a prosphoric choice in the above case gives us:
\begin{gather}
\begin{split}
\Pr\left(P^{(i)}_{jkl}=\left(p^{(i-1)}\right)^+_{jkl}\left|\left(P^{(i-1)}\right)^c_{jkl}=\left(p^{(i-1)}\right)^c_{jkl},u_{a}\right.\right)\\ 
=\frac{\exp\left[u_{a}\left(r\left(\left(p^{(i-1)}\right)^+_{jkl}\right)\right)\right]}{\exp\left[u_{a}\left(r\left(\left(p^{(i-1)}\right)^+_{jkl}\right)\right)\right]+\exp\left[u_{a}\left(r\left(\left(p^{(i-1)}\right)^-_{jkl}\right)\right)\right]}
\end{split}\\
= \ilogit\left[u_{a}\left(r\left(\left(p^{(i-1)}\right)^+_{jkl}\right)\right)-u_{a}\left(r\left(\left(p^{(i-1)}\right)^-_{jkl}\right)\right)\right]. \label{e_choice}
\end{gather}
Thus, $a$'s choice is seen to be governed by the difference in his or her utilities for the states of $Y$ resulting from his or her decision, as desired.

One important special case of the above occurs when there exists some potential function, $\rho$, such that $\rho\left(y^+_{kl}\right) - \rho\left(y^-_{kl}\right) = u_{a}\left(y^+_{kl}\right)-u_{a}\left(y^-_{kl}\right)$ for all $a \in c_{kl}$ and all $\{k,l\}$ (or $(k,l)$ in the directed case).  In such a setting, Equation~\ref{e_choice} reduces to
\begin{equation}
\begin{split}
\Pr\left(P^{(i)}_{jkl}=\left(p^{(i-1)}\right)^+_{jkl}\left|\left(P^{(i-1)}\right)^c_{jkl}=\left(p^{(i-1)}\right)^c_{jkl},\rho\right.\right) =\\
\ilogit\left[\rho\left(r\left(\left(p^{(i-1)}\right)^+_{jkl}\right)\right)-\rho\left(r\left(\left(p^{(i-1)}\right)^-_{jkl}\right)\right)\right]
\end{split}, \label{e_choice2}
\end{equation}
and all decisions depend only on the difference in potentials (which is not agent-specific).  Considered in game theoretic terms, the existence of such a $\rho$ identifies the corresponding behavioral model as a \emph{potential game} \citep{monderer.shapley:geb:1996}, and we likewise say that $\rho$ is a potential for the tuple $(A,\ell,c,\mathcal{Y}_n)$.

\subsection{Equilibrium Distribution of Relations Under the Behavioral Model}

%Need to find cite for independence leading to cross-sectional result
%add \mathcal{P}_n

Given the above, we are now ready to state our primary result:

\begin{theorem} \label{t_eq}
Let $Y$ be the adjacency structure arising from the behavioral model specified by $(\mathcal{Y}_n,A,\ell,c,r,u)$ under edge updating process $X$, and let $Y^{[t]}$ be the state of $Y$ at time $t$.  If $\rho$ is a potential for $(A,\ell,c,\mathcal{Y}_n)$, and $X$ is such that
\begin{enumerate}
\item $X$ is independent of $P$; and
\item $\sum_{x: \tau(x)<t} I\left(a\left(x\right)=i,e_s(x)=j,e_r(x)=k\right) \to \infty$ as $t\to\infty$ a.s. for all $\{j,k\}$ (directed case $(j,k)$) in $E^*(\mathcal{Y}_n)$ and all $i \in c_{jk}$,
\end{enumerate}
then $Y^{[t]}$ converges in distribution to $\Pr\left(Y^{[t]}=y\right)=\left|\{p:r(p)=y\}\right| \frac{\exp\left[\rho\left(y\right)\right)}{\sum_{p' \in \mathcal{P}_n}\exp\left[\rho\left(r\left(p'\right)\right)\right]}$ on support $\mathcal{Y}_n$ as $t \to \infty$.
\end{theorem}
\begin{proof}
We begin by noting that the independence of $X$ and $P$ implies that both $P$ and $Y$ are embedded Markov chains; without loss of generality, then, we may restrict our attention to the sequences $Y^{(1)},\ldots,Y^{(i)},\ldots$ and $P^{(1)},\ldots,P^{(i)},\ldots$ arising from the set of realized edge updating events $x^{(1)},\ldots,x^{(i)},\ldots$.  Since $\rho$ is a potential for $(A,\ell,c,\mathcal{Y}_n)$, we have from Equation~\ref{e_choice2} that  $\Pr\left(P^{(i)}_{jkl}=\left(p^{(i-1)}\right)^+_{jkl}\left|\left(P^{(i-1)}\right)^c_{jkl}=\left(p^{(i-1)}\right)^c_{jkl},\rho\right.\right)$ $= \ilogit\left[\rho\left(r\left(\left(p^{(i-1)}\right)^+_{jkl}\right)\right)-\rho\left(r\left(\left(p^{(i-1)}\right)^-_{jkl}\right)\right)\right]$ (where $j=a(x_i)$, $k=e_s(x_i)$, and $l=e_r(x_i)$).  From Equation~\ref{e_ergcond}, however, this is exactly equal to the conditional probability that $P_{jkl}=p^+_{jkl}$ for a multivariate ERGM process on $P$ with potential $\rho(r(P))$.  Now, consider the sequential edge updating process formed by the realized updating events.  By assumption, $x^{(i)}$ is independent of $P$.  Further, \break $\sum_{x: \tau(x)<t} I\left(a\left(x\right)=i,e_s(x)=j,e_r(x)=k\right) \to \infty$ as $t\to\infty$ a.s. clearly implies that \break $\sum_{i=1}^\infty I\left(a\left(x\right)=i,e_s(x)=j,e_r(x)=k\right) \to \infty$ a.s. as well.  This implies that the sequence $P^{(0)},P^{(1)},\ldots$ satisfies the conditions of Proposition~\ref{prop_gibbs}, and hence $P^{(0)},P^{(1)},\ldots,P^{(i)}$ approaches $\Pr\left(P^{(i)}=p\left|r,\rho\right.\right)=\exp\left[\rho\left(r\left(p\right)\right)\right]/\left[\sum_{p' \in \mathcal{P}_n}\exp\left[\rho\left(r\left(p'\right)\right)\right]\right]$ in distribution as $i \to \infty$.  Since the summation condition on $x^{(i)}$ also implies that the total number of updating events must approach infinity almost surely as $t \to \infty$, this limit must hold in $t$ as well.

It now remains to find the limiting distribution of $Y^{(i)}$.  We proceed by first integrating over $p$, and then substituting from our previous results: 
\begin{align}
\Pr\left(Y^{(i)}=y\left|\rho\right.\right) &= \sum_{p \in \mathcal{P}_n} \Pr\left(Y^{(i)}=y\left|P^{(i)}=p,\rho\right.\right) \Pr\left(P^{(i)}=p\left|r,\rho\right.\right)\\
&= \sum_{p: r(p)=y} \Pr\left(P^{(i)}=p\left|r,\rho\right.\right)\\
&= \sum_{p: r(p)=y} \frac{\exp\left[\rho\left(r\left(p\right)\right)\right]}{\sum_{p' \in \mathcal{P}_n}\exp\left[\rho\left(r\left(p'\right)\right)\right]}\\
&= |\{p: r(p)=y\}| \frac{\exp\left[\rho\left(y\right)\right]}{\sum_{p' \in \mathcal{P}_n}\exp\left[\rho\left(r\left(p'\right)\right)\right]}.
\end{align}
Since this limit also holds in $t$, the limiting distribution of $Y^{[t]}$ must be as derived.
\end{proof}

Theorem~\ref{t_eq} provides us with an expression for the long-run cross-sectional behavior of $Y$, in the sense that a cross-sectional ``snapshot'' of $Y$ taken at a random time will (in the large-$t$ limit) have the specified distribution.  This suggests, in practice, that models for cross-sectional data drawn from sources for which the above assumptions are reasonable may be (in some cases, at least) given a behavioral interpretation.  On the other hand, the result also contains some cautionary elements.  For instance, we note that the limiting distribution of $Y^{[t]}$ is trivially equivalent to a ``typical'' ERGM on $\mathcal{Y}_n$ only in the simple case of unilateral edge control.  In the more general case, $Y^{[t]}$ results from an ERGM mixture over the states of $P$, with properties that depend upon the edge control structure (via $r$).  Thus, distributions arising from relations with multilateral edge control may look very different from those arising from relations with unilateral edge control, \emph{even when the agents' utilities in the two cases are identical.}  One cannot simply pose an arbitrary ERGM and presume that it is a credible equilibrium distribution for a given social process; rather, one must first show that this choice is consistent with the assumed edge control mechanism.

To this latter end, it may be useful to note that the limiting distribution of $Y^{[t]}$ can be rewritten as
\begin{equation}
\Pr\left(Y^{[t]}=y\right)= \frac{ \left|\{p:r(p)=y\}\right| \exp\left[\rho\left(y\right)\right]}{\sum_{y' \in \mathcal{Y}_n} \left|\{p:r(p)=y'\}\right| \exp\left[\rho\left(y'\right)\right]} \label{e_limtilt}
\end{equation}
by grouping like terms within the normalizing factor.  This is immediately recognizable as a standard ERGM on $\mathcal{Y}_n$ with reference measure $h(y) =  \left|\{p:r(p)=y\}\right|$.  In the special case for which $\left|\{p:r(p)=y\}\right|=\left|\{p:r(p)=y'\}\right|$ for all $y,y' \in \mathcal{Y}_n$, the respective weighting factors cancel out, leaving us with a counting measure ERGM($\rho$) on support $\mathcal{Y}_n$.  As noted above, this is trivially true in the unilateral case, for which $|\{p:r(p)=y\}|=1$ for all $y$.  On the other hand, this condition is clearly \emph{not} satisfied for bilateral relations in either the symphonic or epibolic case, as can be appreciated by comparing the set sizes for the undirected graphs $K_n$ (1 and $3^{\tbinom{n}{2}}$, respectively) and $N_n$ ($3^{\tbinom{n}{2}}$ and 1).  Intuitively, symphonic relations allow more opportunities for manifest edges to be absent, while epibolic relations allow more opportunities for manifest edges to be present; this underlying opportunity structure interacts with agent preferences to produce the manifest graph distribution, and must be taken into account.

To get further insight into the limiting distribution of $Y^[t]$, we may further re-write equation~\ref{e_limtilt} by bringing the reference measure within the exponentiated potential, i.e.
\begin{align}
\Pr\left(Y^{[t]}=y\right) &= \frac{ \exp\left[\ln \left|\{p:r(p)=y\}\right| + \rho\left(y\right)\right]}{\sum_{y' \in \mathcal{Y}_n}  \exp\left[\ln \left|\{p:r(p)=y'\}\right| + \rho\left(y'\right)\right]}\\
&= \frac{ \exp\left[\ln h(y) + \rho\left(y\right)\right]}{\sum_{y' \in \mathcal{Y}_n}  \exp\left[\ln h(y') + \rho\left(y'\right)\right]}, \label{e_limfin}
\end{align}
with reference measure $h(y)=\left|\{p:r(p)=y\}\right|$ being the number of ways in which $y$ can arise under $r$, a quantity corresponding to the \emph{multiplicity} of $y$.  Written in this form it is apparent that the distribution of $Y[t]$ can be viewed as a standard ERGM whose graph potential has two additive components: the behavioral potential, $\rho$, based on agents' utilities, and a deterministic offset, $\ln h$, that depends only on the way in which agents' choices are resolved into edges. This separation of effects greatly facilitates prediction of the impact of $r$ on network structure, since $h$ typically has a very simple form.  For instance, consider the family of cases in which $r$ is edgewise decomposible, in the sense that the state of the $(i,j)$ edge (i.e., $r(p)_{ij}$)  depends only on $p_{\cdot jk}$.  This family includes all scenarios in which decisions are not inherently ``bundled'' (i.e., one can attempt to create or sever one edge without creating or severing another), as is clearly the case for most social science applications.  For this family, we can further decompose $s$ in terms of its edgewise components, i.e. $h(y)=\prod_{ij}h'_{ij}(y_{ij})$; where the same rules apply for every edge (i.e., where $r$ is both edgewise decomposible and homogeneous), this further simplifies to $h(y)=\prod_{ij}h'(y_{ij})$, which (gathering terms) can be reduced simply to $h'(1)^{\sum_{ij}y_{ij}} h'(0)^{\sum_{ij}(1-y_{ij})}$.  Working with this last expression, it is helpful to observe that the deterministic relationship between edges and non-edges means that both counts need not be maintained:
\begin{align*}
h(y)&=h'(1)^{\sum_{ij}y_{ij}} h'(0)^{\sum_{ij}(1-y_{ij})}\\
&= \left(\frac{h'(1)}{h'(0)}\right)^{\sum_{ij}y_{ij} h'(0)^{\sum_{ij}1}}\\
&\propto \left(\frac{h'(1)}{h'(0)}\right)^{\sum_{ij}y_{ij}}.
\end{align*}
Since the constant of proportionality here is invariant to $y$, it divides out of equation~\ref{e_limfin} and can be ignored.  We are left with the graph potential element $\ln h(y) = \left(\sum_{ij}y_{ij}\right) \ln \left(h'(1)/h'(0)\right)$ which is immediately recognizable as an offset to the edge parameter of the combined ERGM.  This quantity is easily calculated: for instance, in the bilateral symphonic case $h'(1)=1$ and $h'(0)=3$ (epibolic case, $h'(1)=3, h'(0)=1$) and hence the effect on the combined model is to shift the effective edge parameter by $\ln (1/3) \approx -1.1$ (epibolic case, $\ln 3 \approx 1.1$).  That something as fundamental as changing the norms by which an interaction process operates can be expressed in a deterministic offset is both gratifying and useful, in that it allows us to immediately model the impact of a hypothetical shift in $r$.  It also provides an effective foundation for inference regarding $\rho$: given an observation $y$ from the equilibrium distribution of $Y[t]$, we need only estimate the combined graph potential using standard ERGM inference methods \citep[e.g.,][]{schweinberger.et.al:ss:2020} and then adjust the resulting parameter estimates to account for the effect of $r$.  While inhomogeneous and/or non-decomposible resolution functions may result in a more complex adjustment, this basic approach can be employed with any choice of $r$.

In addition to adjustment for tie formation norms, Theorem~\ref{t_eq} incorporates the assumption that  equilibrium potential $\rho$ is also a potential for the utilities of the agents in $A$ (i.e., that the associated choice process is a potential game).  For this condition to be satisfied, it is clearly the case that $\rho$ must arise from a behaviorally credible choice of $u$.  In the Section~\ref{sec_application}, we examine this issue more closely, focusing on strategies for constructing plausible potentials and on their subsequent use in inference.

\subsubsection{Opportunity, Entropy, and Relational Norms}

The fact that relational norms (via $r$) enter into our equilibrium ERGM via the log multiplicity immediately suggests an entropic interpretation of the associated term: the hidden dynamics of the choice process create more \emph{opportunities} for some patterns of relations to occur, which manifests as an entropic force biasing network structure.  As noted e.g. by \citet{butts:jms:2019}, the log of the reference measure, $h$ can in general be interpreted as the entropy function for the graph microstate associate with hidden degrees of freedom (i.e., unobserved generative processes).  Under the counting measure ($h(y)\propto 1$), every distinct network can arise in one way, and the only factors influencing network structure are those arising from the social forces expressed by $\theta^T t(y)$.  As we have seen, this is compatible with a network arising from a decision process with unilateral edge control.  Other generative processes, however, lead to different reference measures: under the contact formation process (CFP) with uniform local population density, for instance, one obtains the Krivitsky reference \citep{krivitsky.et.al:statm:2011} $h(y)=n^{-t_e(y)}$ (where $t_e$ is the edge count), which favors sparse graphs; this reflects the act that dense graphs under the CFP require many local encounters among individuals, and there is a lower density of trajectories with many such encounters than with few.  Here, we can directly observe how latent degrees of freedom emerge from multilateral edge control, and how the translation of decision patterns into realized edge states via relational norms leads to differences in entropy by graph state.  This in turn allows us to, given a fitted or posited model, make arguments about how the resulting networks would be expected to change under different relational norms (holding preferences constant).  We illustrate the use of this idea below.

\section{Application to Behavioral Inference} \label{sec_application}

%basic idea: posit utilities, derive potentials, fit models, compare models

Theorem~\ref{t_eq} provides a formal basis for moving from a choice process to a cross-sectional model in ERGM form.  In practice, however, our interest is more often in the inverse problem: given cross-sectional network data, we would like to infer something about the behavioral process which generated it.  Where the background assumptions of the previous section are met, this type of behavioral inference is often possible.  For instance, different agent utility functions typically lead to distinct equilibrium distributions -- since the equilibrium distribution serves as a likelihood for $Y^{[t]}$, comparison of competing utility functions vis a vis some observed cross-section $y$ may be reduced to a standard problem of likelihood-based model selection.  In other cases, agent utilities may be posited in terms of an unknown parameter vector (e.g., expressing the relative valuation of different positional properties); here, standard parameter estimation methods may be used to determine (with some degree of of uncertainty) these unknown quantities.

In this section, we discuss the construction of candidate utility functions for the purpose of behavioral inference.  Following this, we illustrate the use of these principles with a simple example involving friendship in a professional organization, showing how we can use the behavioral interpretation to gain additional insights from an empirically calibrated model.

\subsection{Constructing Behaviorally Interpretable Potentials}

The basic procedure for constructing behaviorally interpretable potentials is in principle straightforward.  First, one identifies a hypothetical utility function (or linearly separable element thereof) for some agent $a_i$, $u_i$.  Next, one identifies the difference in $u_i$ resulting from a single edge change for which $i$ is a controller (i.e., $u_i\left(y^+_{jk}\right)-u_i\left(y^-_{jk}\right)$.  Finally, one seeks some function $\rho$ such that $\rho\left(y^+_{jk}\right)-\rho\left(y^-_{jk}\right)=u_i\left(y^+_{jk}\right)-u_i\left(y^-_{jk}\right)$ for all $y,j,k$, and for all $i \in c_{jk}$.   For complex utility functions, it is useful to observe that if $u_i\left(y\right)=\sum_l u_{il}\left(y\right)$ (i.e., if $i$'s utility is linearly decomposible) and if $\rho_l\left(y^+_{jk}\right)-\rho_l\left(y^-_{jk}\right)=u_{il}\left(y^+_{jk}\right)-u_{il}\left(y^-_{jk}\right)$ under the above conditions, then $\rho\left(y\right)=\sum_j \rho^{(j)}_i\left(y\right)$ is a potential for $u_i$ (and for the entire system, if all conditions are met).  Thus, where partial utilities are linearly separable, one can construct a global potential by deriving and aggregating partial potentials.  

It is noteworthy that many of the statistics used in ERGM parameterization (and notably those arising from Hammersley-Clifford \citep{pattison.robins:sm:2002}) are aggregates of indicators for the presence or absence of local subgraphs or induced subgraphs (graphlets); these lend themselves to the above treatment, provided that we can approximate the partial utility associated with the presence or absence of the subgraph (e.g., a specific, labeled two-star) to be equal for all controllers of the edges in the subgraph.  (Note that total utility functions need not be the same across agents, nor across different labeled subgraphs.)  When controllers' partial utilities are not equal, the equilibrium (and thus coefficients estimated by fitting models to data) will reflect a weighted average of the controllers' utilities.

\subsection{Example: Friendship in a Law Firm}

Here, we demonstrate behavioral interpretation of ERGMs by application to a well-known empirical study of friendship in a law firm \citep{lazega1992analyse}.  We begin by fitting a cross-sectional model to the network, with terms motivated by basic social mechanisms.  We then provide a behavioral interpretation of the resulting parameter estimates, and leverage the generative nature of the ERGM process to perform counterfactual analyses of the inferred social system.  Counterfactual analysis enables us to explore the implications of our model beyond the insights offered by conventional null hypothesis tests, most particularly the potential impact of a hypothetical change in social norms regarding the way that friendship is defined.  

\subsubsection{Lazega Friendship Network}

Members of groups and organizations have preferences that influence their behavior, impacting their decisions on what they want to do and who they want to do it with. We consider an instance of such a group, namely a friendship network within a professional setting collected by \citet{lazega1992analyse}.  The network was collected within a law firm between 1988-1991, with the firm having offices in three locations in the northeastern US (Boston, Hartford, and Providence).  The nodes in the network (also here taken to be the agents) are the 71 firm members, for which we have information on age, gender (male or female), the law school a particular member attended, their office, which field of law they work in (litigation or corporate practice), their seniority (number of years at the firm), and their status (partner or associate). There is also dyadic covariate information on whether or not each pair of individuals are coworkers.

%It consists of 854
%directed edges and 71 nodes. 

Relational data was collected by self-reports of perceived friendship; as the specific question used implies interactions that are mutual by definition,\footnote{As reported by Lazega, ``Quels sont parmi vos coll\`{e}gues ceux avec lesquels vous avez des activit\'{e}s sociales dehors du travail, des activit\'{e}s qui n'ont rien \`{a} voir avec le travail? Par exemple des coll\`{e}gues dont vous connaissez bien la famille et qui connaissent la v\^{o}tre'' \citep[p583]{lazega1992analyse}; roughly, ``Which are the co-workers with whom you have social activities outside of work, activities that have nothing to do with work?  For example, co-workers whose families you know well and who know yours?'' Though reported in French, the original instrument was presumably in English.} we interpret the data to consider of multiply reported undirected ties (with two reports per edge).  Following the error analysis study of \citet{lee.butts:sn:2018}, we employ the Intersection LAS \citep{krackhardt:sn:1988} as our estimate of the underlying network; that is, we take a relationship to be present between two individuals if and only if both report that the relationship is present.  After integrating informant reports, we observe a relatively sparse graph with a mean degree of 3.63 (density 0.0519) and a moderate level of degree skewness (degree centralization 0.14).  The network is visualized in Figure~\ref{fig:lazega_orig}.  As can be seen, the network is comprised of two large and somewhat cohesive subgroups, each largely organized around a particular office.  In addition to a number of individuals who are pendant to each group, we observe a relatively large number of isolates (14 of 71), who appear not to have personal friendships within the organization.

%% Much of this section hinges on what we select as the final model,
%% so it's more skeleton-like right now.

\begin{figure}
  \begin{center}
    \includegraphics[width=0.5\textwidth]{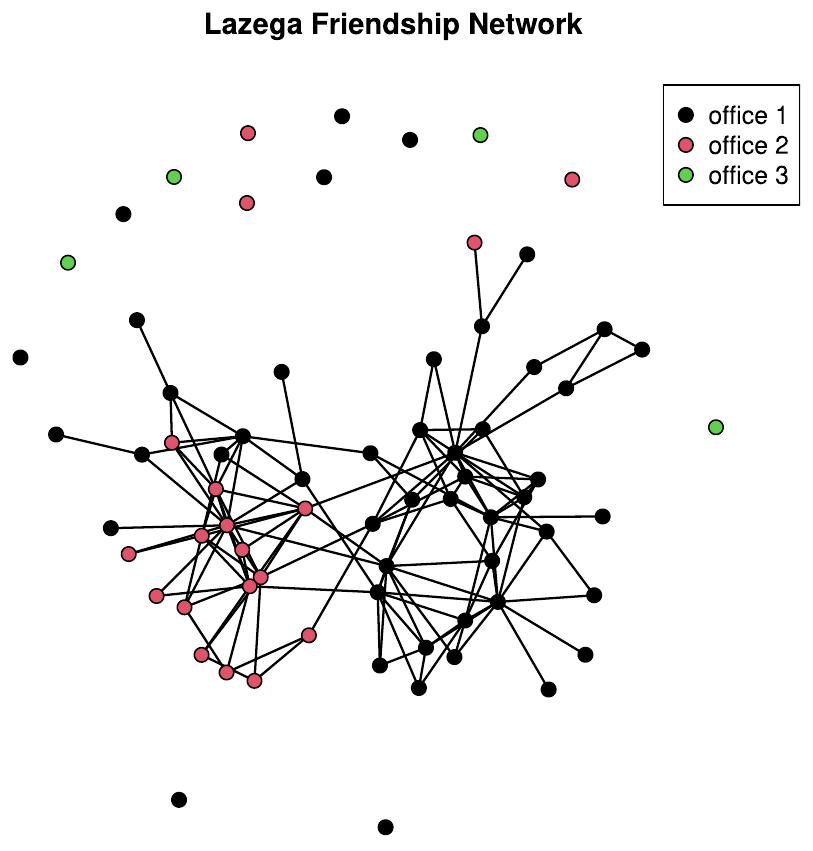}
    \caption{The Lazega friendship network; vertex colors indicate office.   \label{fig:lazega_orig}}
  \end{center}
\end{figure}

\subsubsection{Model Specification and Fitting}

Given that our relation of interest clearly involves a bilateral investment of time and interest, we interpret this network as a symphonic relation under bilateral control, with the controllers of each edge being the agents associated with the two endpoints.  We here propose a number of mechanisms that could potentially affect agents' decisions regarding the friendship network; each is here interpreted in terms of its hypothetical contribution to the partial utilities of its edge controllers.

\begin{list}{}{}
\item \emph{Baseline partial utilities:} Individuals may clearly vary in their baseline preferences for within-firm friendship.  An overall baseline partial utility can be obtained via a standard edgecount term (which acts as an intercept within an ERGM), while per-agent differences can be inferred by adding sender/receiver fixed effects for each agent (\emph{sociality terms}).  To identify the model, we take the partial friendship utility of the first agent as a reference.  The resulting model parameters indicate the baseline utility of each agent to add or maintain a friendship, prior to consideration of other factors.
\item \emph{Mixing preferences:} Homophily preferences are well-known in friendship formation \citep{mcpherson.et.al:ars:2001}, whether driven by the relative ease of interacting with others like oneself \citep{mayhew.et.al:sf:1995} or by other considerations; in mileu in which personal relationships are pursued for strategic ends, by contrast, ties to those differently situated from oneself may provide greater value \citep{granovetter:ajs:1973}.  In both cases, however, agents can be expected to care about the relationships between their own attributes or subgroup memberships and those of potential alters.  Here, the partial utilities for such effects (i.e., the difference in utility for forming a relationship due to differential attributes, net of other factors) can be accomplished by covariate effects for matching on key variables (nodematch statistics), and absolute differences along other dimensions (absdiff statistics).  Here, we examine matching for gender, status in the firm, law school attended, office, and type of law practiced.  Absolute differences are employed for age and seniority within the firm.  The resulting model parameters indicate the utility difference associated with adding a tie to an alter matching on the indicated characteristic (respectively, differing on the characteristic by one unit), net of other factors.
\item \emph{Engagement preferences:} In addition to an overall utility for forming ties, individuals may have other preferences regarding their overall level of engagement with the network.  These can include preferences for or against isolation, or concurrency (maintaining at least two ties, e.g. so as to avoid the liabilities of pendancy \citep{yamagishi.cook:spq:1993}).  We capture these respective effects by terms for counts of isolates and concurrent vertices.  For the isolate statistic, the associated parameter can be viewed as the utility \emph{penalty} for adding an edge to an alter for an ego who is currently isolated; a negative effect hence suggests that being an isolate is aversive.  For the concurrent vertices statistic, the associated parameter indicates the excess utility obtained by forming an additional edge, when one is currently a pendant (i.e., becoming a concurrent vertex).  Obviously, these same effects also apply (\emph{mutatis mutandis}) for decisions to remove an edge.
\item \emph{Embeddedness preferences:} Finally, we observe that individuals have long been observed to be concerned with the embeddedness of their relations within broader substructures that may either facilitate or complicate individual relationships.  In the context of the law firm, one such context is co-work: direct co-workers may be convenient alters with whom to form friendships, our could alternately be seen as strategically risky given professional obligations.  We can probe this by using the network of co-work ties as an edgewise covariate predicting friendship.  Co-work also serves as a site for observation of behavior, and we might hypothesize that when two co-workers have common friends who are also co-workers with both of them, that this encourages formation of a friendship tie within the focal dyad.  We capture this effect by a local triangle term (in the sense of \citet{strauss.ikeda:jasa:1990}), where closure effects are limited to co-work cliques.  As type of legal practice also creates a potential basis for common ground, we further consider cyclic ties (i.e., relationships embedded within cliques) among those with common practice as an additional effect.  Finally, we consider a partial utility effect for positive-tie balance, representing a general preference to both close and to avoid creating open triangles.  The parameters in these cases all indicate the utility contribution to ego for ties that maintain or increase the associated type of embeddedness, net of other factors.
\end{list}

Model estimation for ERGMs with per-node sociality parameters is difficult, requiring regularization for reasonable performance; likewise, we have a large number of hypothesized terms, not all of which may in fact shape preferences.  To both stabilize the sociality terms and shrink non-predictive effects towards zero, we employ L2 regularized estimation (analogous to ridge regression).  We perform hyperparameter tuning using an approximate inference strategy based on maximum pseudo-likelihood inference, with 10-fold cross-validation on the set of edge variables employed to tune the regularization parameter (with a pseudo-likelihood objective; optimization by golden section search).  Regularization has been shown to improve MPLE performance \citep{vanduijn.et.al:sn:2009}, and the efficiency of the procedure allows us to employ it where computational costs would otherwise be prohibitive.  Final estimates were then obtained using L2-regularized stochastic approximation, employing a regularized variant of the algorithm of \citet{snijders:joss:2002}; $5 \times 10^4$ burn-in iterations with a thinning interval of $5 \times 10^5$ were employed for all draws.  Visualization and analysis was performed in the \texttt{R} statistical computing system \citep{rteam:sw:2026} using custom scripts employing functions from the \texttt{sna} \citep{butts:jss:2008b}, \texttt{network} \citep{butts:jss:2008a}, and \texttt{ergm} \citep{hunter.et.al:jss:2008,krivitsky.et.al:jss:2023} packages of the \texttt{statnet} \citep{handcock.et.al:jss:2008} library.

\begin{figure}
  \begin{center}
    \includegraphics[width=.75\textwidth]{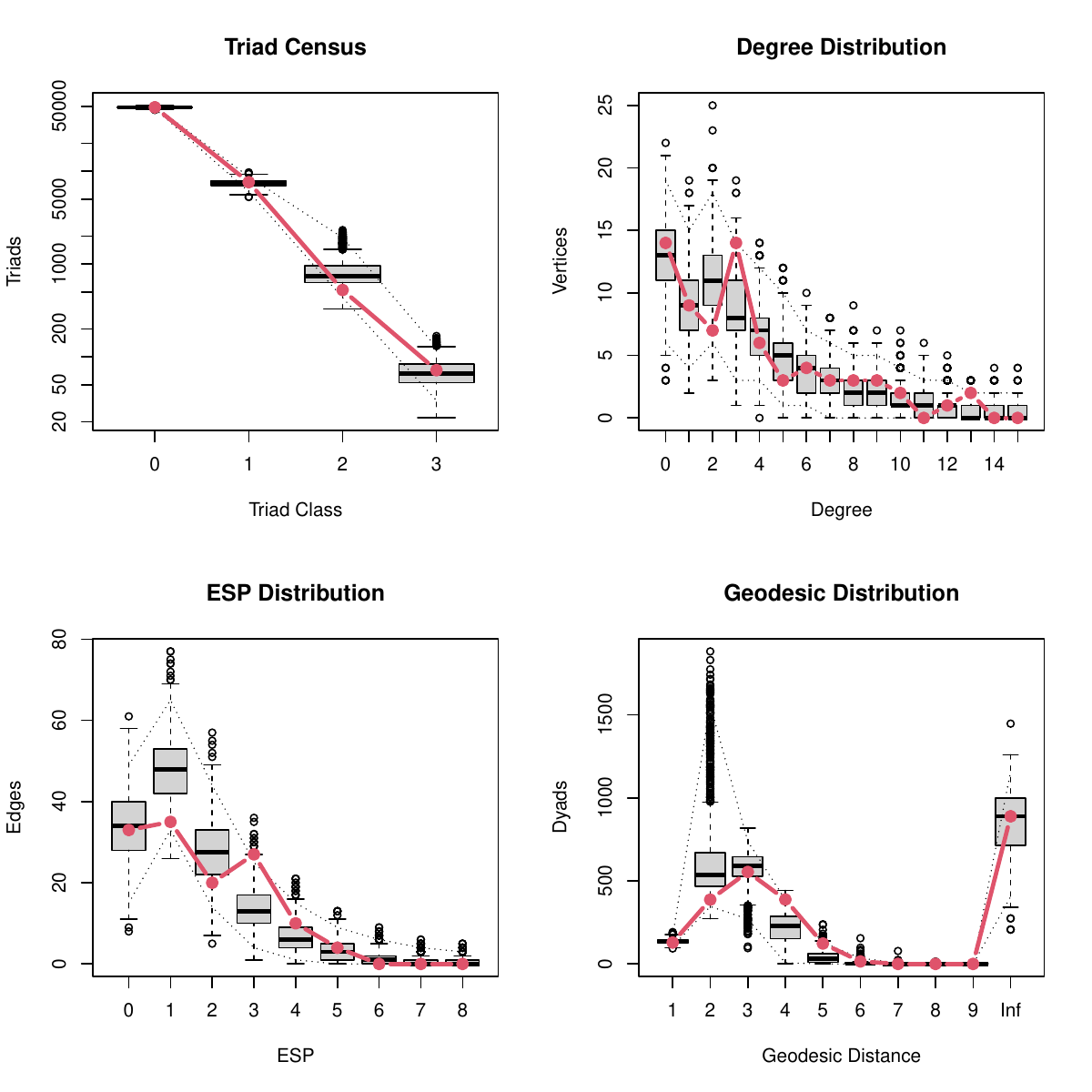}
    \caption{Model adequacy checks for the Lazega model; dotted lines indicate 95\% simulation intervals, red points indicate observed data.  \label{fig:adequacy}}
  \end{center}
\end{figure}

We assess model adequacy using the standard approach \citep{goodreau.et.al:jss:2008} of simulating draws from the fitted model and verifying coverage of observed network properties by the resulting simulation intervals (1000 draws, 16,384 burn-in iterations, thinning interval 1,024).   The results are shown in Figure~\ref{fig:adequacy}.  The fitted model successfully recapitulates the overall network structure, and we thus regard it as adequate for our subsequent analyses.

\subsubsection{Parameters and Interpretation}

Parameter estimates for the friendship network model are shown in Tables~\ref{tab:reg_eff} and \ref{tab:soc_eff}; due to the large number of sociality effects, these are shown distributionally in Figure~\ref{fig:soc_eff}.  Note that, due to L2 regularization, standard error and $p$-value estimates are heuristic. The strongest systematic effects are associated with embeddedness, with friendship ties to co-workers having higher payoff, ties to co-workers with shared co-working friends being highly favorable, and with an additional tendency to form clique-embedded ties with alters who share common law practice.  There is also a tendency to form ties with members of the same office or status. Net of these strong effects, we do not see reliable effects of other attributes (though the signs of almost all estimates are consistent with homophily preferences).

\begin{table}[ht]
\centering
\begin{tabular}{lrrrl}
  \hline\hline
  Term & $\hat{\theta}$ & se & $p$-value &  \\ 
  \hline
  edges & -2.34 & 1.21 & 0.05 & *\\ 
  isolates & -0.58 & 0.80 & 0.47 &\\ 
  concurrent vertices & 0.59 & 0.64 & 0.36 &\\ 
  local co-work triangles & 0.17 & 0.34 & 0.62 &\\ 
  co-practice cyclic ties & 0.56 & 0.24 & 0.02 &*\\ 
  co-work ties & 1.57 & 0.36 & 0.00 & ***\\ 
  balance & -0.05 & 0.01 & 0.00 & ***\\ 
  nodematch (gender) & -0.09 & 0.32 & 0.79 &\\ 
  nodematch (status) & 0.71 & 0.33 & 0.03 &*\\ 
  nodematch (law school) & -0.04 & 0.25 & 0.89 &\\ 
  nodematch (office 1) & 0.82 & 0.99 & 0.41 &\\ 
  nodematch (office 2) & 2.51 & 1.14 & 0.03 &*\\ 
  nodematch (practice) & 0.69 & 0.39 & 0.08 &\\ 
  absdiff (age) & -0.01 & 0.03 & 0.79 &\\ 
  absdiff (seniority) & -0.01 & 0.04 & 0.75 &\\ 

  \hline\hline
\end{tabular}
\caption{Fitted model coefficients for all non-sociality effects.%; zero coefficients indicate parameters that have been regularized to below 0.005 in magnitude. 
 L2 regularization parameter $\lambda=0.03998$.} 
\label{tab:reg_eff}
\end{table}

\begin{figure}
  \begin{center}
    \includegraphics[width=.5\textwidth]{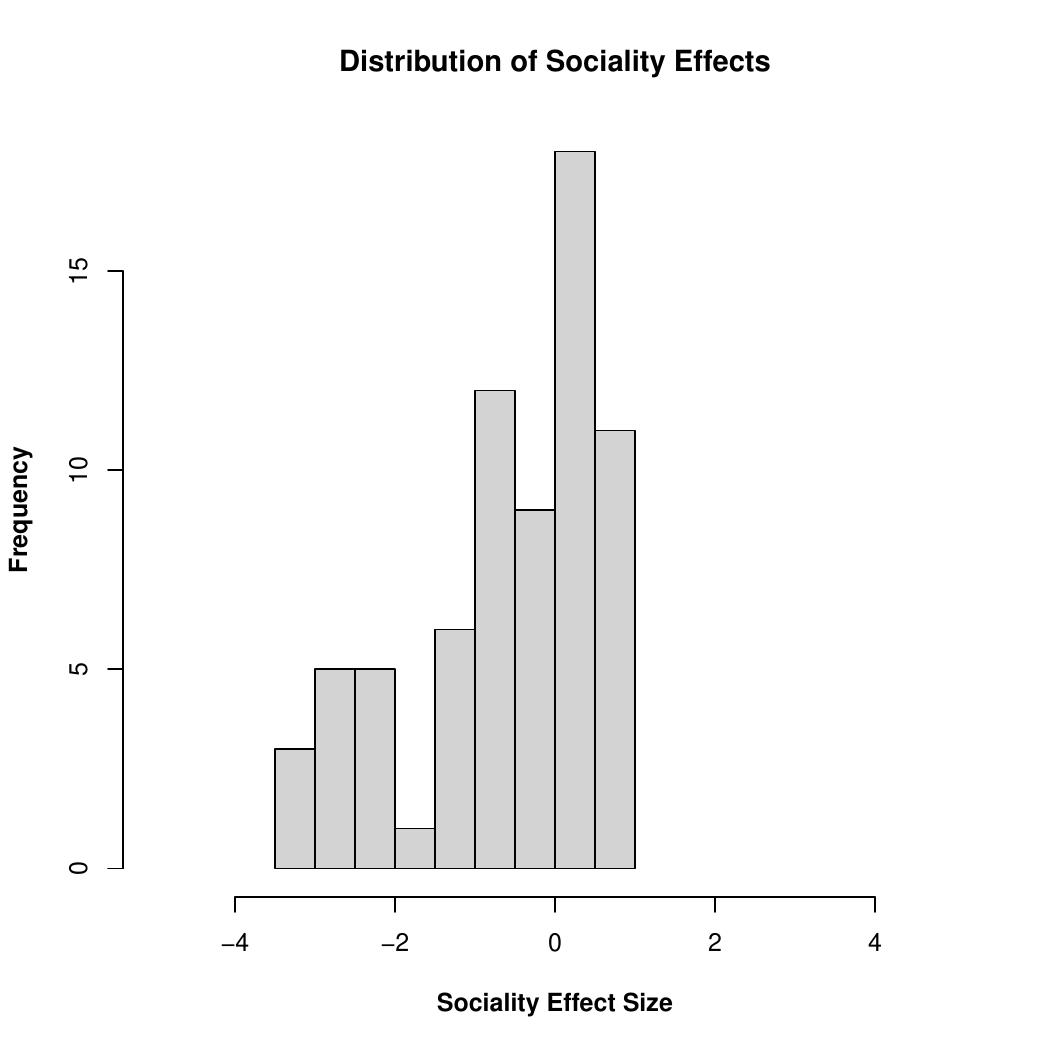}
    \caption{The distribution of sociality effects in the fitted model.  \label{fig:soc_eff}}
  \end{center}
\end{figure}

Although the individual coefficients are not precisely estimated, the overall distribution of sociality parameters in Figure~\ref{fig:soc_eff} provides some insight into preference heterogeneity within the network.  Broadly, the vast majority of agents have fairly similar baseline utilities for within-firm friendship (as evidenced by the strong peak at zero), but a fair number do appear to find friendships more or respectively less attractive than average.  This suggests the possibility of differing preferences for social interaction amongst the agents within the network.

\begin{table}[ht]
\centering
\begin{tabular}{rrrrrrrrrrrr}
  \hline\hline
  Term	& $\hat{\theta}$	& se	& $p$-value&	Term	& $\hat{\theta}$	& se	& $p$-value	&Term	& $\hat{\theta}$	& se	& $p$-value \\
  \hline
  sociality2	&	-2.58	&	2.75	&	0.35	&	sociality27	&	-0.31	&	1.2	&	0.79	&	sociality52	&	0.23	&	1.08	&	0.83	\\
  sociality3	&	-3.13	&	2.09	&	0.14	&	sociality28	&	0.73	&	0.5	&	0.14	&	sociality53	&	-0.77	&	1.35	&	0.57	\\
  sociality4	&	-1.23	&	1.45	&	0.4	&	sociality29	&	0.32	&	1.07	&	0.76	&	sociality54	&	0.53	&	0.99	&	0.6	\\
  sociality5	&	0.39	&	0.62	&	0.53	&	sociality30	&	0.55	&	0.5	&	0.28	&	sociality55	&	0.07	&	0.99	&	0.94	\\
  sociality6	&	-0.93	&	0.94	&	0.32	&	sociality31	&	-0.04	&	0.55	&	0.94	&	sociality56	&	0.68	&	0.98	&	0.49	\\
  sociality7	&	-1.16	&	1.13	&	0.31	&	sociality32	&	-0.1	&	0.54	&	0.85	&	sociality57	&	0.43	&	1.05	&	0.68	\\
  sociality8	&	-2.68	&	3.04	&	0.38	&	sociality33	&	-1.52	&	0.99	&	0.12	&	sociality58	&	-0.97	&	0.78	&	0.21	\\
  sociality9	&	-2.61	&	2.58	&	0.31	&	sociality34	&	-1.1	&	1.55	&	0.48	&	sociality59	&	-0.79	&	0.86	&	0.36	\\
  sociality10	&	-2.24	&	2.23	&	0.32	&	sociality35	&	0.62	&	0.57	&	0.27	&	sociality60	&	0.67	&	1.02	&	0.51	\\
  sociality11	&	0.4	&	0.99	&	0.68	&	sociality36	&	-2.32	&	1.94	&	0.23	&	sociality61	&	-0.72	&	1.59	&	0.65	\\
  sociality12	&	0.44	&	0.99	&	0.66	&	sociality37	&	-2.35	&	3.64	&	0.52	&	sociality62	&	0.35	&	1.12	&	0.75	\\
  sociality13	&	0.81	&	0.92	&	0.38	&	sociality38	&	-0.46	&	1.16	&	0.69	&	sociality63	&	-1.11	&	0.96	&	0.25	\\
  sociality14	&	-3.47	&	1.94	&	0.07	&	sociality39	&	0.65	&	0.98	&	0.51	&	sociality64	&	0.49	&	1.1	&	0.65	\\
  sociality15	&	-2.62	&	2.36	&	0.27	&	sociality40	&	0.43	&	0.98	&	0.66	&	sociality65	&	0.26	&	1.03	&	0.8	\\
  sociality16	&	-0.22	&	1.13	&	0.85	&	sociality41	&	0.45	&	0.98	&	0.65	&	sociality66	&	-0.86	&	1.21	&	0.48	\\
  sociality17	&	-1.04	&	1.34	&	0.44	&	sociality42	&	0.2	&	1.11	&	0.85	&	sociality67	&	-0.12	&	1.14	&	0.92	\\
  sociality18	&	-3.21	&	1.47	&	0.03	&	sociality43	&	0.12	&	1.04	&	0.91	&	sociality68	&	-0.54	&	1.22	&	0.66	\\
  sociality19	&	0.62	&	0.97	&	0.52	&	sociality44	&	-2.32	&	3.15	&	0.46	&	sociality69	&	-0.05	&	1.09	&	0.96	\\
  sociality20	&	0.14	&	1.11	&	0.9	&	sociality45	&	0.21	&	1.11	&	0.85	&	sociality70	&	-0.88	&	1.45	&	0.55	\\
  sociality21	&	0.63	&	0.98	&	0.52	&	sociality46	&	-0.03	&	0.82	&	0.97	&	sociality71	&	-0.37	&	1.28	&	0.77	\\
  sociality22	&	-0.54	&	1.15	&	0.64	&	sociality47	&	-2.61	&	3.46	&	0.45	&		&		&		&		\\
  sociality23	&	-2.43	&	2.04	&	0.23	&	sociality48	&	-0.87	&	1.5	&	0.56	&		&		&		&		\\
  sociality24	&	0.63	&	0.96	&	0.51	&	sociality49	&	-0.86	&	1.29	&	0.51	&		&		&		&		\\
  sociality25	&	-1.3	&	1.17	&	0.26	&	sociality50	&	0.28	&	0.75	&	0.71	&		&		&		&		\\
  sociality26	&	0.2	&	0.96	&	0.84	&	sociality51	&	-0.82	&	0.87	&	0.34	&		&		&		&		\\

   \hline\hline
\end{tabular}
\caption{Fitted model coefficents for sociality effects; L2 regularization parameter $\lambda=0.03998$.} 
\label{tab:soc_eff}
\end{table}

\subsubsection{Counterfactual Analysis}

One advantage of the generative, ERGM framework is the ability to gain insights by performing \emph{counterfactual analyses,} in which we perturb particular aspects of the model from their calibrated values and examine the resulting impact on social structure (an approach long used in agent-based modeling \citep{macy.willer:ars:2002}). Here, we provide some simple examples of such computational experiments.

We begin with the role of sociality effects.  These represent heterogeneity in individual friendship preferences, and given the relatively low level of observed variation, one might expect this heterogeneity to have a minimal impact on structure.  On the other hand, heterogeneity in incentives to pursue ties can potentially have non-trivial structural effects: for instance, because those with more ties are \emph{ceteris paribus} more likely to become tied to others with more ties, degree heterogeneity tends to create core-periphery structure, which in turn affects phenomena ranging from information diffusion to social influence \citep{friedkin:bk:1998}.  Although this impact of high degree preference is known, to what extent are those who do \emph{not} seek friendship critical to network structure?  We can probe this effect of removing below-reference preference variation by simply setting all negative sociality parameters to 0, and simulating from the resulting model.  While we would expect to see a \emph{somewhat} higher density under these conditions, the results are stark.  A typical realization is shown in the left-hand panel of Figure~\ref{fig:soc_eff_perturb}: we obtain an extremely dense, clique-like structure that is very far from what is empirically observed (or, as we can see from Figure~\ref{fig:adequacy}, is created from our fitted model).  

\begin{figure}
  \begin{center}
    \includegraphics[width=.45\textwidth]{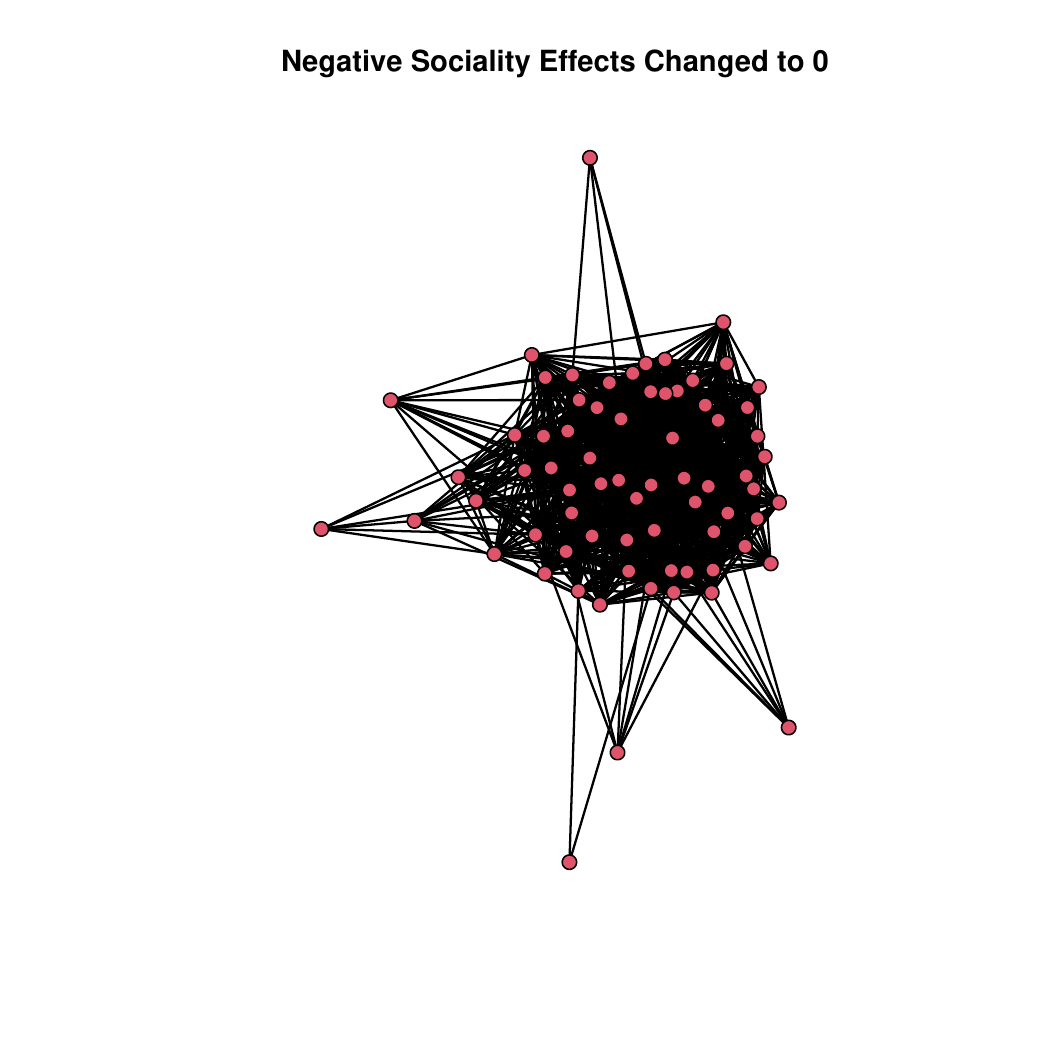} \includegraphics[width=.45\textwidth]{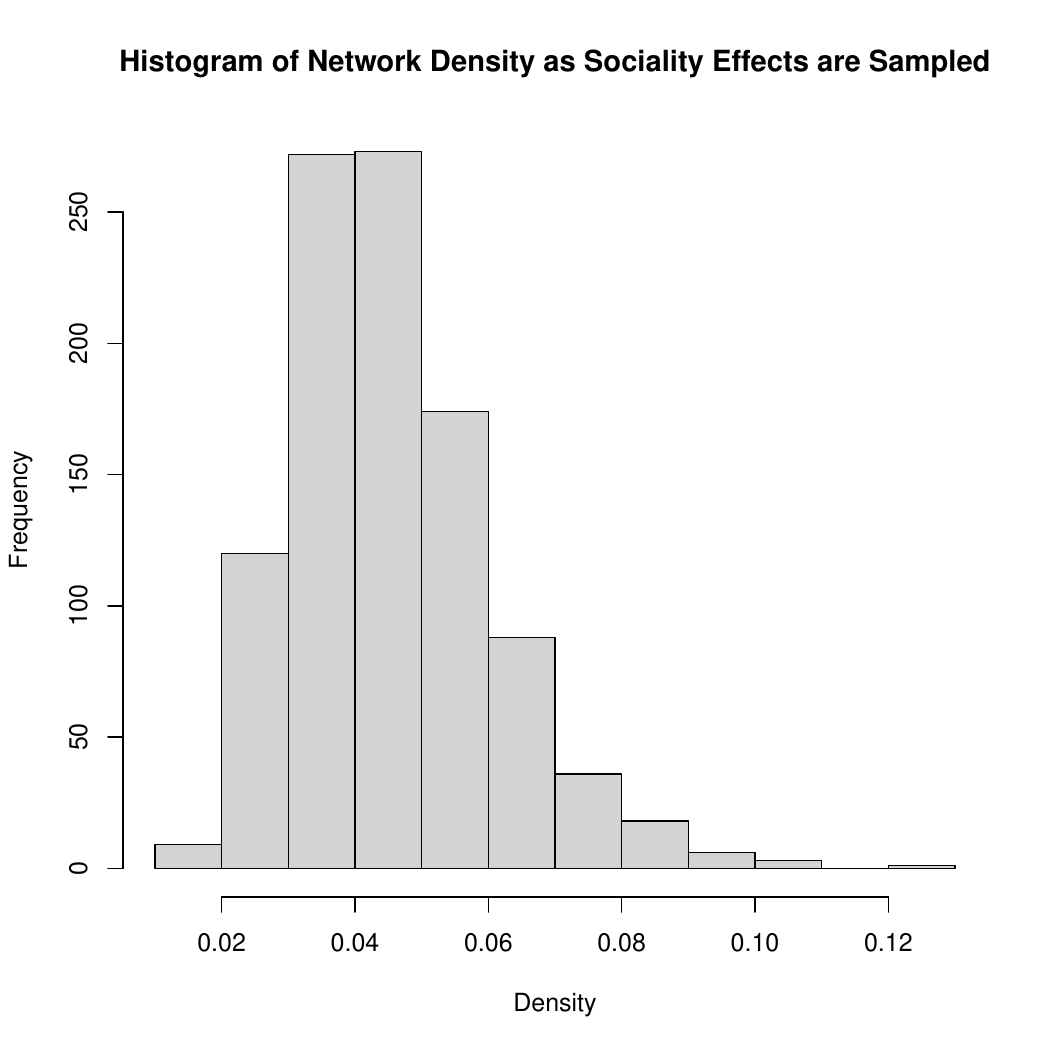}
    \caption{(Left) Representative draw from the friendship model with negative sociality terms removed.  ``Neutralizing'' those who place a lower value on friendship leads to radically increased levels of density over the entire network. (Right) The network density distribution under permuted sociality parameters (1,000 replicates).  While the alignment of friendship preferences with attributes can have some effect, densities remain relatively close to the observed level over all permutations. \label{fig:soc_eff_perturb}}
  \end{center}
\end{figure}

%\begin{figure}
%  \begin{center}
%    \includegraphics[width=.5\textwidth]{soc_eff_neg_to_0.pdf}
%    \caption{Representative draw from the friendship model with negative sociality terms removed.  ``Neutralizing'' those who place a lower value on friendship leads to radically increased levels of density over the entire network. \label{fig:soc_eff_neg_to_0}}
%  \end{center}
%\end{figure}

%\begin{figure}
%  \begin{center}
%    \includegraphics[width=.5\textwidth]{soc_eff_robustness.pdf}
%    \caption{ The density distribution of networks predicted by our
%      model when the sociality effects are permuted (1,000 replicates).  While the alignment of friendship preferences with attributes can have some effect, densities remain relatively close to the observed level over all permutations. \label{fig:soc_robust}}
%  \end{center}
%\end{figure}

We can thus see that those with low levels of interest in friendship play an important role in preventing density explosion.  Is this a generic effect, or does it depend upon a particular alignment between overall preferences for friendship and e.g. other attributes?  We probe this question by randomly permuting the sociality coefficients (holding all other factors equal), and simulating from the resulting models. The density distribution of the resulting networks is shown in the right-hand panel of Figure~\ref{fig:soc_eff_perturb}.  While the graph density varies across permutations, it is centered near the observed value of 0.0456, and at its most extreme does not exceed approximately .125 (for a mean degree of approximately 8.9).  This implies that, while fortuitous alignments of friendship preference and other attributes would be expected to increase or decrease density to a minor extent, heterogeneity per se is playing a role in preventing convergence to a high-density regime.  Taken together, it appears to be the case that having a significant fraction of nodes with relatively low levels of interest in forming at-work friendship ties acts as a ``buffer'' that prevents runaway growth of tight-knit friendship groups, and avoids a dramatic refactoring of network structure.  This effect of heterogeneity has not been well-studied, and is an interesting target for further investigation.

\begin{figure}
  \begin{center}
    \includegraphics[width=.45\textwidth]{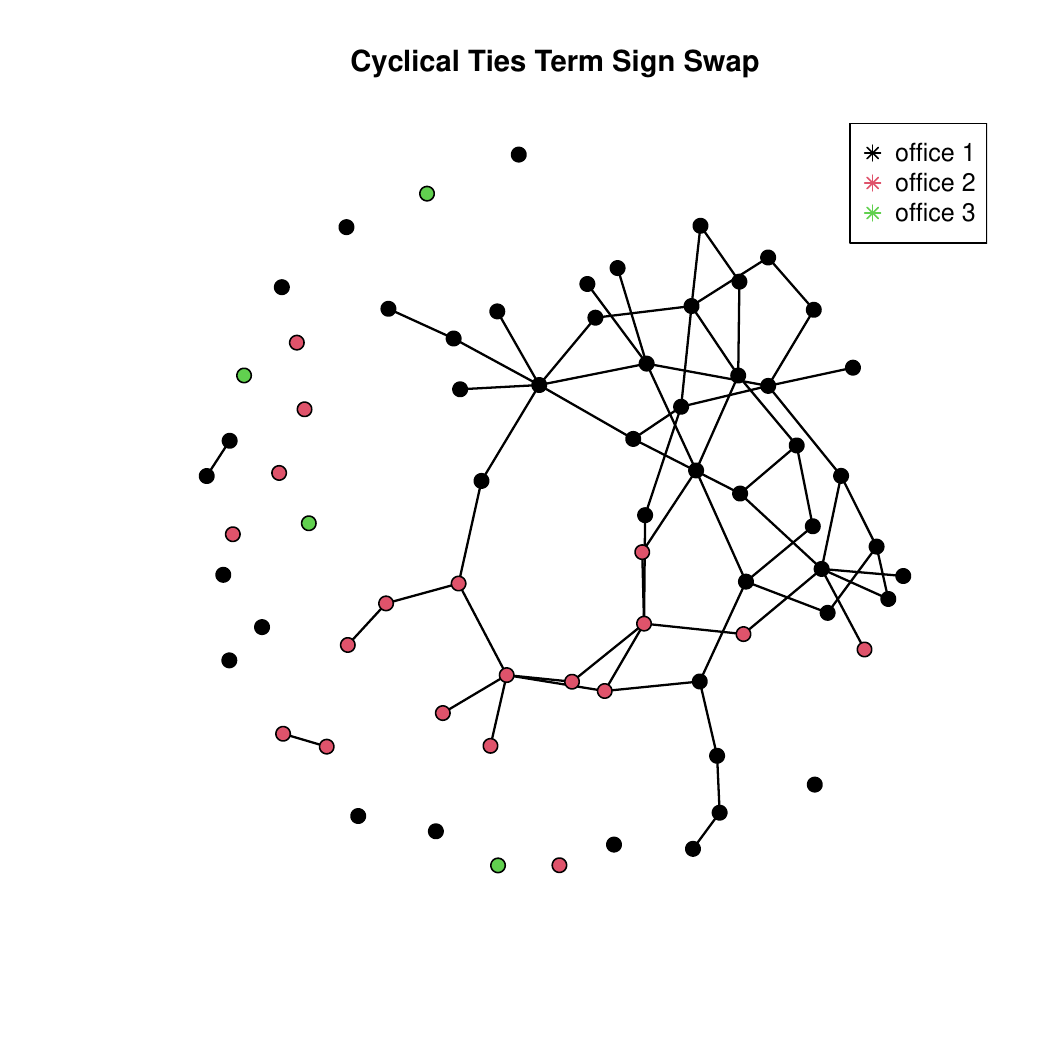} \includegraphics[width=.45\textwidth]{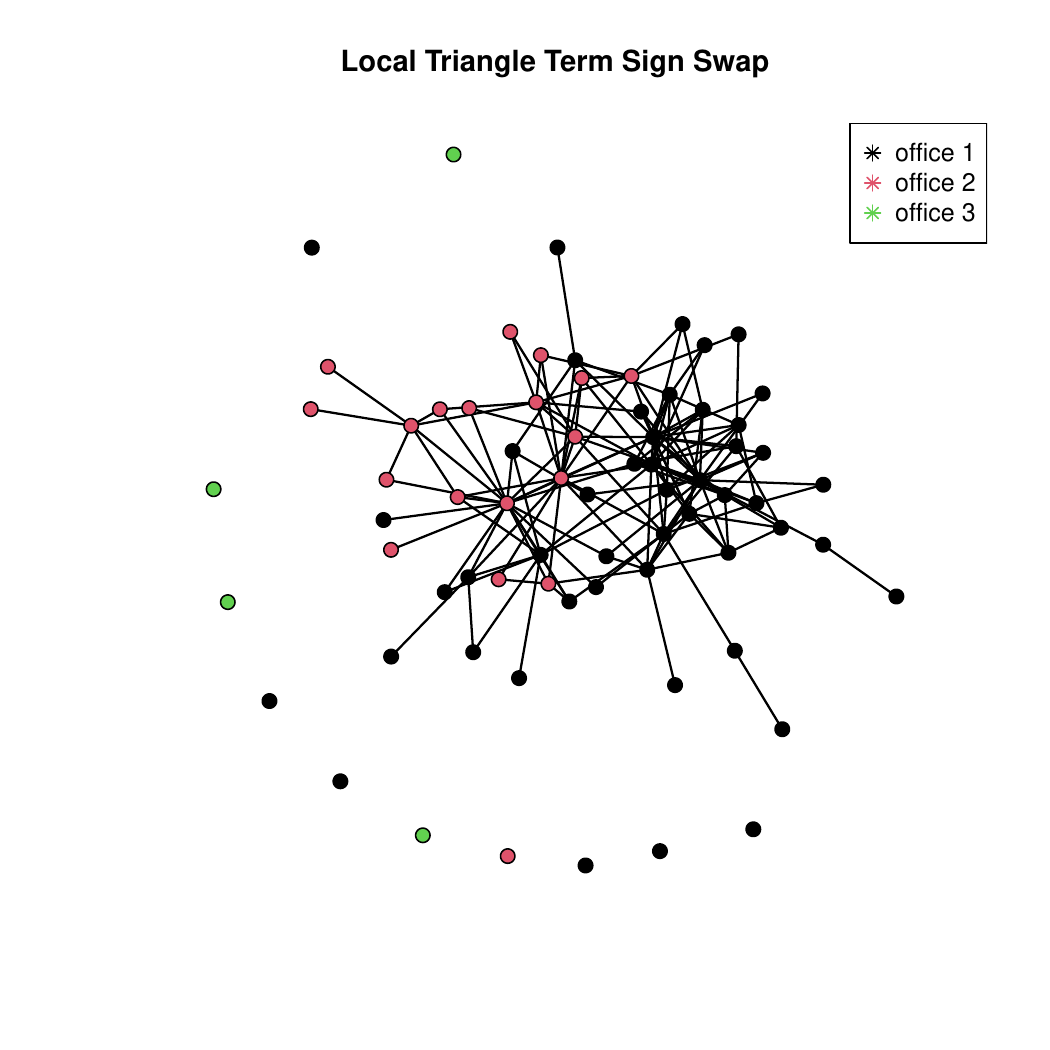}
    \caption{ Typical realizations for networks under the model with sign flipped cyclical tie (left) and local triangulation (right) parameters. Suppression of within-practice cyclic ties greatly suppresses local density, while local trangulation suppression has a weaker effect.   \label{fig:cyc_flip}}
  \end{center}
\end{figure}

While degree preference heterogeneity may play an important buffering role, closure effects are the major drivers of densification and the creation of local cohesion.  Within the friendship network model, this manifests in two significant effects: a preference for formation of ties to those in the same type of legal practice that are embedded in triangles; and an increasing preference for ties with co-workers as a function of the number of co-workers who are mutual friends.  Here, too, we can gain insights into the role of these partial utilities by perturbing the fitted model and examining the impact of perturbation of closure effects on network structure.  Figure~\ref{fig:cyc_flip} shows typical realizations for networks generated under two such perturbations: changing the sign of the cyclic tie parameter (left); and changing the sign of the local triangulation parameter (right).  In each case, we switch from a partial utility that favors closure to one that disfavors it; the consequences, however are interestingly distinct.  Disfavoring tie closure among those practicing the same type of law does not greatly reduce global connectivity within the network, but ``hollows out'' local clique structure, leaving a giant component that is largely biconnected via large cycles with few chords, decorated with small pendant trees.  Comparison with Figure~\ref{fig:lazega_orig} underscores the differences.  We can thus see that considerable ``in-filling'' of locally cohesive structure in the friendship network seems to be impacted by this effect.  By contrast, the impact of co-work based local triangulation is more subtle: the right-hand panel of Figure~\ref{fig:cyc_flip} indeed looks quite similar to the original graph.  Reversing the tendency for mutual co-working friends to incentivize closure has at best a localized density suppression effect, and does not radically alter network structure.  Thus, what may seem like very similar types of closure effects may have very different structural consequences, and hypothetical preference shifts based on different types of shared contexts (similar type of work versus working together) may not lead to similar outcomes.

\begin{figure}
  \begin{center}
    \includegraphics[width=.6\textwidth]{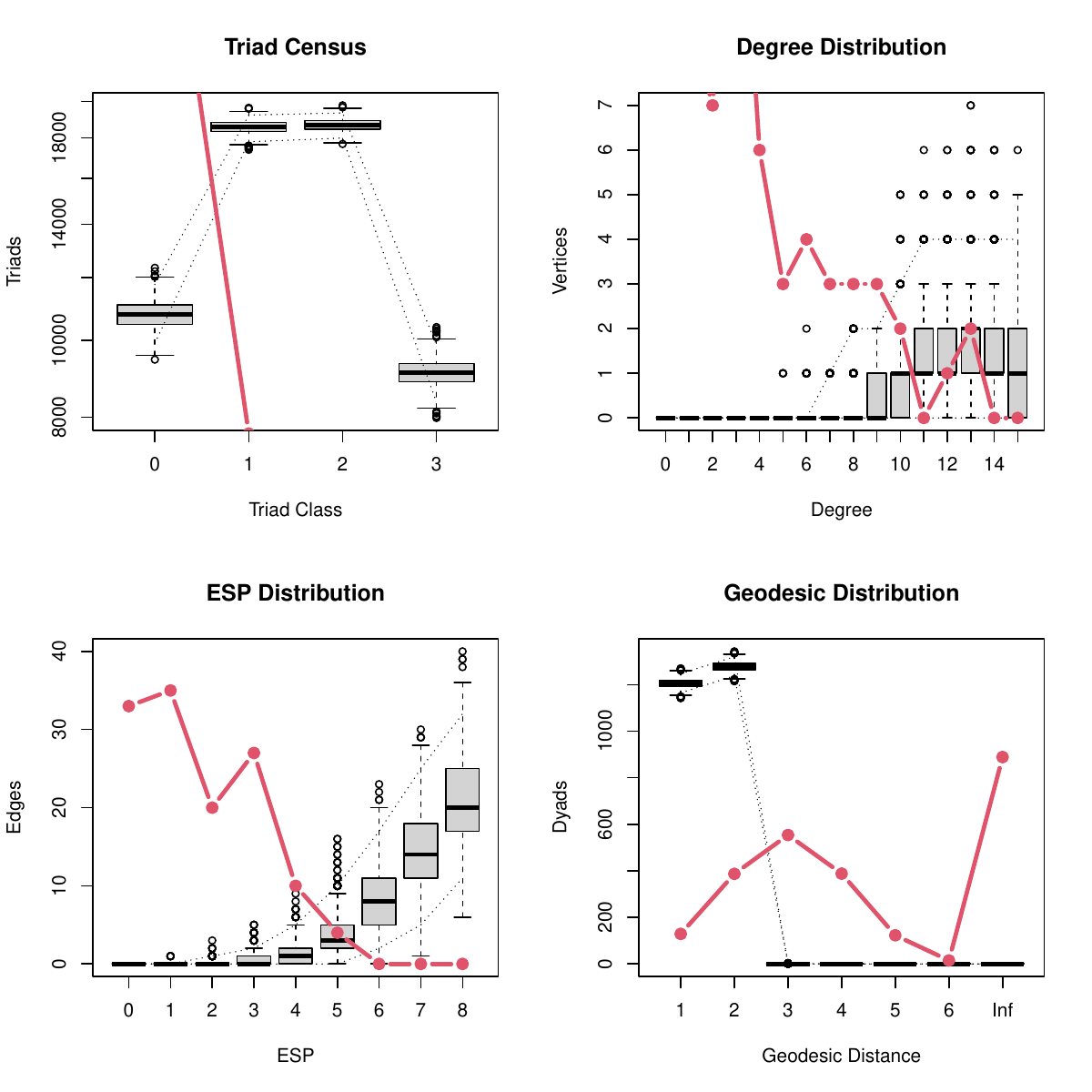}
    \caption{Consequences of shifting relational norms from symphonic to epibolic, while holding preferences constant.  The resulting network properties (represented by 95\% simulation intervals, 1,000 simulated replicates) represent a radically different social order from that observed (red). \label{fig:laz_normshift}}
  \end{center}
\end{figure}

The notion of a preference shift associated with one reference group versus another also raises the question of how more fundamental changes in relational norms would be expected to affect the network if preferences were held constant.  The Lazega friendship network is clearly symphonic; however, we can easily explore the consequences of equivalent preferences under an epibolic resolution function by simulating draws from an equivalent model with an appropriately shifted edge parameter.  Figure~\ref{fig:laz_normshift} shows the resulting impact on typical network properties, in the same form as the symphonic model of Figure~\ref{fig:adequacy}; as before, the observed network statistics are given in red, for reference.  Not surprisingly, changing the social ``rules'' in such a way as to allow anyone to impose friendship on anyone else increases mean degree.  It also greatly increases the depth of local triangulation (as seen in the ESP distribution), and reduces the median diameter of the network.  Plainly, such a network would not be sustainable as ``friendship'' in the sense that it would have been experienced by the members of Lazega's study community.  This, in turn, suggests how the mechanisms of normative relational definition may be maintained: frowning on imposition, and requiring mutuality to accept interactions keeps the network in check in ways that maintain its value.  By turns, if it were impossible to defend the symphonic norm (as e.g., in an online community in which technical changes make it impossible to screen out unwanted social interactions), it is plausible that the resulting structural changes would drive changes in relational preferences, which would themselves keep the network in check (or possibly force it into a new regime entirely).  Our model, calibrated on this single observation, cannot speak directly to how those changes might occur.  However, being able to investigate the immediate consequences of either preference or norm changes allows us to see which changes could be reasonably accommodated without radically altering social structure, and which changes would be expected to lead to changes that could potentially prove broadly destabilizing.  As we show below, when and whether norm shifts are destabilizing can depend on where one is in the space of potential preference functions, leading to structural effects that can be contingent on preference in a complex way.

\section{Application to Theoretical Studies of Small Group Structure} \label{sec_phase}

Small groups often exist in stable configurations for extended periods of time, before experiencing rapid change. For instance, members of a small hobby group or student organization may have regular, cohesive patterns of interactions that persist for months or even years, before suddenly collapsing in the face of seemingly minor changes in social conditions. An open source project might have a stable group of software developers who communicate regularly for many years, before a fragmenting into several (possibly competing) projects.  Such rapid shifts between radically different structural regimes can be viewed as phase transitions.  At one point a subject of considerable sociological interest (see e.g. \citet{fararo:bs:1978,fararo:bs:1984}), studies of phase behavior arguably waned; recent work on phase transitions in ERGMs \citep[e.g.][]{haggstrom.jonasson:jap:1999,radin.yin:aap:2013,grazioli.et.al:jpcB:2019,butts:jms:2021}, however, has once more brought them to the fore.  Ironically, some of this more recent work has been motivated by concerns about model specifications that produce to unwanted phase transitions (generally under the rubric of \emph{degeneracy}; see \citet{handcock:ch:2003,schweinberger:jasa:2011,butts:sm:2011b}), leading to a general view of such behavior as innately pathological.  However, the discovery of non-pathological phase transitions in both human and non-human networks has re-opened the notion of complex phase behavior as a potentially real and important phenomenon to be captured (rather than as a problem to be avoided).  Arguably, this is an especially natural perspective in a small group context, where radical shifts in group structure are not uncommon.  Groups coalesce from previously non-interacting members, later splintering again.  Extreme structures such as cliques, stars, or empty graphs can emerge and persist within small groups, unlike most large networks (the term ``clique'' itself being borrowed from a common small-group structure).  This suggests that ERGMs may be particularly useful as tools to study these types of emergent phenomena despite their historical lack of use in this area (though see \citet{von.et.al:sn:2021}).

Here, we show two examples of how the framework developed here can be used to examine phase behavior in small groups.  In particular, we leverage our ability to examine norm shifts in bilaterally controlled relations (i.e., where a symphonic relation is made epibolic or vice versa) to see how preferences and norms interact to stabilize or destablize particular types of structures.

\subsection{A Simple Model of Coalescence and De-coalescence Under Norm Shifts}

In small groups, it can be common for weak relations reflecting simple interaction or communication to be strongly cohesive (to the point of being complete, or nearly complete).  Indeed, such high-density interaction may be critical to group solidarity (and, for teams, performance).  For non-human social animals that form \emph{fission-fusion societies,} coalescence into high-density groups is part of a regular (usually diurnal) cycle, and can be essential for survival.  It is thus interesting to consider transitions between dense and sparse phases in small group settings.

While there are many mechanisms that can induce such phase transitions, we focus here on a particularly simple class (studied e.g. by \citet{handcock:ch:2003}).  Consider an undirected bilateral relation on a small group of $n$ nodes, where $V=A$, and where agent utilities have two components: a base cost of maintaining each edge; and a \emph{group synergy} that results from leveraging interactions among partners.  In particular, we presume that the utility change for agents $i$ and $j$ in adding a new $\{i,j\}$ tie is equal to $\theta_e + \theta_s (d_i + d_j)$, where $d_i$ and $d_j$ are the respective degrees of $i$ and $j$ prior to the addition.  This second term can be viewed as a first-order (i.e., linear) approximation to the synergy effect, with each existing tie increasing the potential value of adding a new tie.  The graph potential implementing such a utility function is $\rho(y) = \theta_e t_e(y) + \theta_s t_{2s}(y) + \log h(y)$, where $t_e$ is the edge count and $t_{2s}$ is the count of 2-stars (subgraphs in which one node is adjacent to two others).  Under epibolic edge formation norms (i.e., if either party can impose the relation on the other), $h(y) = \log(3) t_e(y)$, while symphonic norms imply $h(y)= - \log(3) t_e(y)$.  For our purposes, it is particularly interesting to consider the regime where interaction is costly, but synergies are present (e.g., $\theta_e<0$, $\theta_s>0$).

The edge/two-star model is known to undergo a density explosion phase transition \citep{handcock:ch:2003} in which, for sufficiently large $\theta_s$ (or $n$), the gains to group interaction overwhelm the costs, and the group converges to a highly dense phase.  Below this threshold, interaction costs dominate, and the relation is extremely sparse. To demonstrate this, we draw parameters for both $t_e$ and $t_{2s}$ from a uniform distribution ranging from -5 to 5, then calculate the density of a network simulated using the parameterized model. This procedure is then repeated 100,000 times. The left-hand panel of Figure~\ref{fig:small_group_test} illustrates the resulting density distribution, for a small ($n=7$) epibolic network.  %% \ctb{Describe the simulations, in enough detail that they could be repeated. -- Added. is 2 sentences above.}
As can be seen, for every partial edge utility, there exists a sharp transition in the partial utility of 2-stars ($\theta_s$) such that graphs above the critical value are nearly always complete (or nearly complete), and those below the value are nearly empty.  For such groups to remain cohesive, it is thus clear that they must remain above the transition line; even very small preference changes that take the group below the critical threshold will destabilize the group structure, leading to dissolution.

\begin{figure}
  \begin{center}
    \includegraphics[width=1\textwidth]{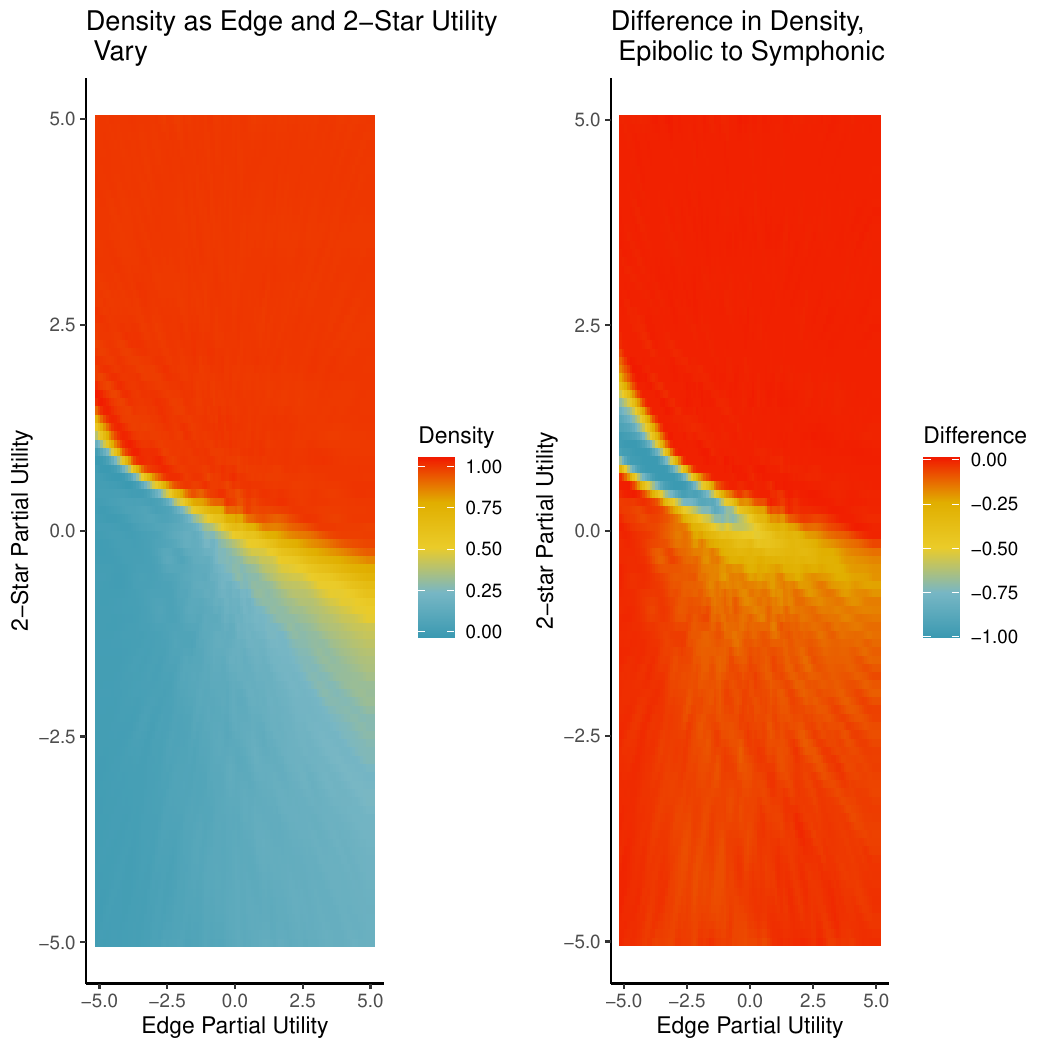}
    \caption{A simple simulation of a potential game ERGM, where the
      network is determined by two terms: A 2-star term and an edge
      term, both allowed to vary. Edge control is bilateral and
      epibolic (imposed by either party). The left plot displays how
      network density changes as 2-star and edge utility change under
      an epibolic regime. The right plot displays the difference in
      density surface as we move from an epibolic to a symphonic relational norm. %% \ctb{Title in the figure doesn't seem to be right - if density is going down, we are shifting from epibolic to symphonic. -- Should be fixed now}
      \label{fig:small_group_test}}
  \end{center}
\end{figure}

If small changes in preferences can have a radical effect, what of edge formation norms?  From the above, we can see that switching from epibolic to symphonic edge resolution has the impact of decreasing the edge parameter by $2 \log 3 \approx 2.2$, potentially altering the stability of the system.  The impact of this change for a network of size $n=7$ is shown in the right-hand panel of Figure~\ref{fig:small_group_test}.  For much of the preference space, there is little effect: the network remains either stably cohesive or stably empty under a norm shift.  However, as shown, there is a large diagonal band within the preference space where previously stable groups are disrupted by the norm shift.  

Although an intendedly simplistic model, this provides a number of insights.  It demonstrates that, in populations of groups, small differences in either base tie costs or the gains to group interaction can lead to large differences in group cohesion.  For currently stable (or unstable) groups near the phase transition, small shifts in preferences my lead to radical changes in group behavior; at the same time, however, quite large preferences changes may have little effect for groups far from the transition boundary.  Perhaps most interestingly, norm shifts may radically alter group stability, but only for groups that lie within a particular subspace of the preference landscape.  If such a shift occurs over a large population of groups (whether due to changes in interaction technology, policy, or culture), it may thus heavily impact some groups while leaving other, seemingly similar groups unchanged.

\subsection{Hierarchy Stabilization}

As a second example, we turn to a case studied by \citet{yu.et.al:siam:2021} on stabilization of extremely centralized group structures.  This case was motivated by studies of (often pathological) groups controlled by a central leader or leadership team, in which the leader interacts with all members but discourages them from having strong side interactions that are unmediated by him or her.  (Such groups are sometimes referred to popularly as ``cults,'' though this term is disfavored by many researchers.)  Yu et al. posit a minimal model for sustaining ``cult''-like star structures, analytically deriving conditions under which such stars are \emph{locally stable} under this model family (i.e., no network formed by a single edge change is more favorable than the star itself).  Put in behavioral terms, we may consider the Yu et al. model to involve an undirected, symmetric relation with $V=A$, in which agent utilities have two contributions: a base edge cost, and a penalty for being part of a null relation that is unbridged by some third party.  This last can be understood as a ``distrust of strangers'' condition, in which group members find it aversive to be in a group with someone who is neither a direct contact, nor is at least bridged by a shared partner.  A group leader seeking to maintain their position may plausibly encourage such distrust, while also raising the cost of interaction; done properly, this may allow him or her to maintain his or her position for long periods of time.

\begin{figure}
  \begin{center}
    \includegraphics[width=1\textwidth]{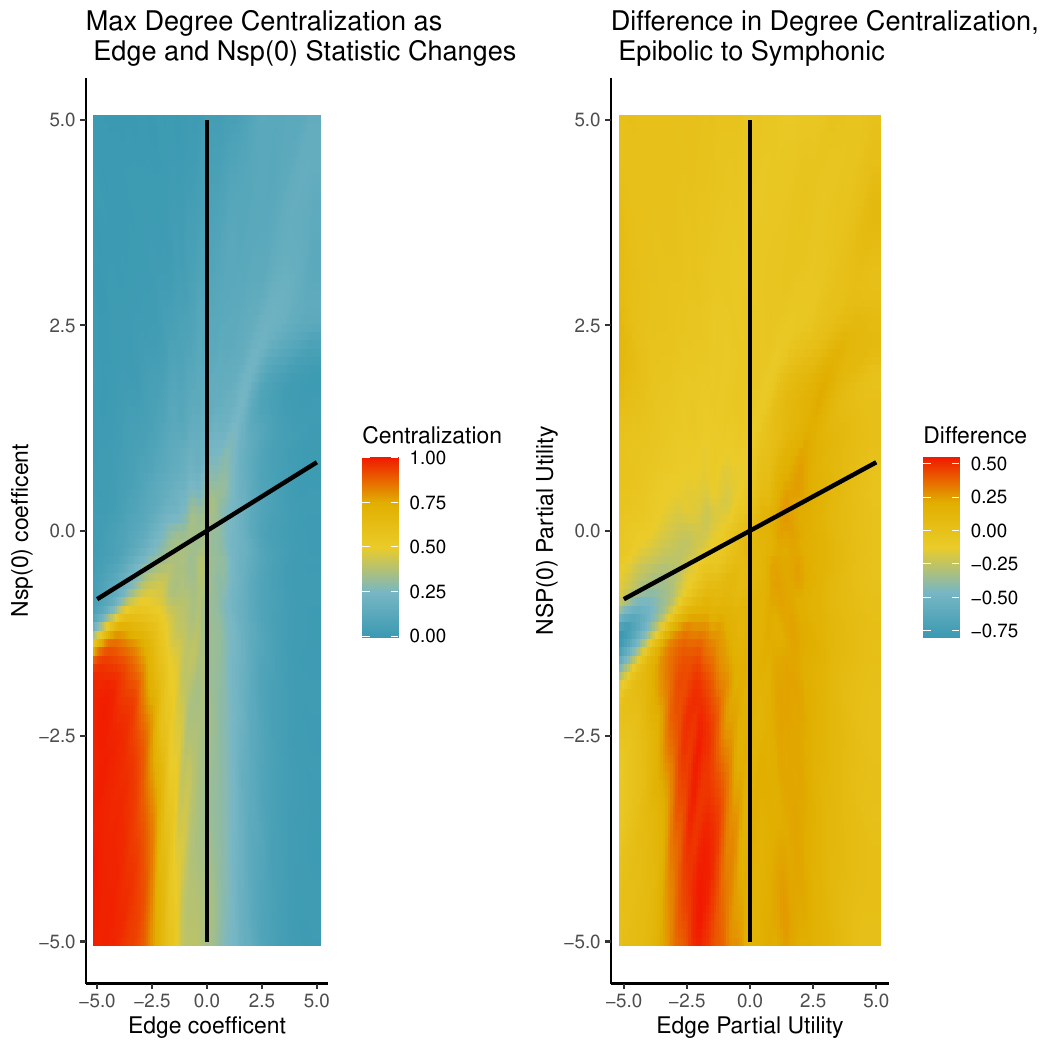}
    \caption{ (left) Centralization by partial utilities in the epibolic cult model; locally stable region for the star network indicated by black lines (lower left quadrant).  A sharp transition separates star-stabilizing preferences from those leading to unstructured or non-cohesive groups.  (right) Change in centralization when passing from epibolic to symphonic norms.  Groups in certain regions of preference space experience either radical enhancement or reduction of centralized structures, while others are unaffected.  \label{fig:cult_test}}
\end{center}
\end{figure}

The ERGM form that implements the above model has two terms: $t_e$, and $t_{NSP0}$ (the former being the edge count, and the latter being the count of null dyads with no shared partners in common (NSP(0))  \citet{yu.et.al:siam:2021} show that the locally stable regime arises as a subset of the parameter space in which both parameters are negative.  To investigate the behavior of this model, we draw both $t_e$ and $t_{NSP0}$ from a uniform distribution ranging from -5 to 5, then calculate the centralization of a network simulated using the parameterized model. This procedure is then repeated 100,000 times%% .\ctb{describe simulation study in enough detail for replication-- should be done (is sentence before thie)}
.  The behavior of the model in the epibolic case is shown in the left-hand panel of Figure~\ref{fig:cult_test}.  Although less sharp than the cohesion case above, we see a relatively clean transition between a highly centralized phase and a highly decentralized phase (with the former residing inside the stable cone).  We can thus see that the success of a hypothesized leader in maintaining centralization depends on preventing behavior from crossing the phase boundary, with the likely impact being sudden decentralization.  Figure~\ref{fig:cult_examples} shows examples of typical realizations from each quadrant of the preference space; transitions that make edge formation more favorable lead to dense, decentralized groups, while those that make NSP(0)s more favorable without enhancing edge favorability lead to empty graphs (and, likely, the dissolution of the group).  Different types of preference shifts thus lead do different failure modes (at least, from the perspective of the leader).

As with the cohesion case, it is also interesting to consider the impact of norm shifts on this network.  This is shown in the right-hand panel of Figure~\ref{fig:cult_test}.  Like the edge/2-star family, norm shifts edge/NSP(0) family have little effect over much of the preference space.  However, there are exceptions: we see a large vertical band (in red) where switching to a symphonic norm increases centralization, marking a potentially vulnerable region that includes a very wide range of partial utilities for NSP(0), but a relatively narrow range for edge costs.  We also see a smaller region of centralization loss, present near the edge of the stable cone.  It is thus possible for the same norm shift to have different directions of effect on different groups, depending on where they lie in the preference space.

%% The results in Figure \ref{fig:cult_test} suggest that under a
%% bilateral symphonic (mutual agreement) control regime, a stable cult
%% is only likely to exist when in a single region of model
%% coefficients. When we change the control regime to one where the cult
%% leader can forbid certain edges from existing (blending a bilateral
%% symphonic rule with an epibolic one), the size of this region
%% increases, as the two well connected but decentralized regions become
%% more centralized (but less connected) due to the cult leader vetoing
%% edges between regular members.

\begin{figure}
\begin{center}
    \includegraphics[width=1\textwidth]{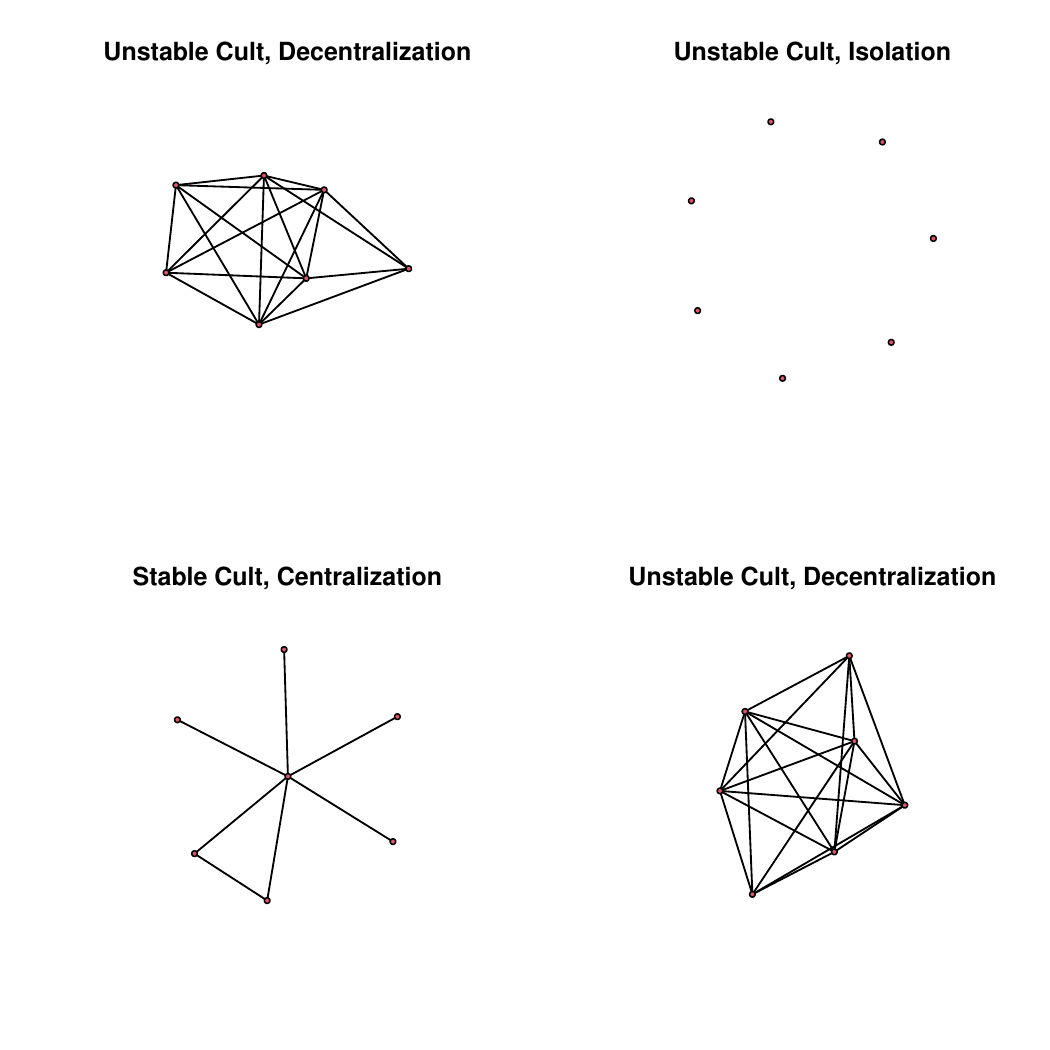}
    \caption{Four example networks produced by different pairs of
      coefficients in Figure \ref{fig:cult_test}. In order from
      left-to-right, $(-2.5, 2.5)$, $(2.5, 2.5)$, $(-2.5, -2.5)$,
      $(2.5, -2.5)$. Each of these coefficients are characteristic of
      their respective ``region'' (separated from one another by black
      lines). Two of these sets of coefficients result in an unstable
      and decentralized cult, where the leader is unable to force
      members to reach one another through the leader. One pair of
      coefficients results in a ``cult'' with few to no ties between
      members (a state of isolation). Only one set of coefficients
      results in a stable cult, where the leader is central, and most
      members must go through the leader to reach other
      members. \label{fig:cult_examples}}
\end{center}
\end{figure}

%SO MANY THINGS WE DON'T HAVE TIME TO TALK ABOUT.... -CTB
%\section{Formal Relationships}
%
%\subsection{Utility, Energy, and Temperature}
%
%\subsection{Degeneracy and Nash Equilibria}
%
%\subsection{Utility Functions and Dependence Graphs}

\section{Conclusion} \label{sec_conclusion}

In this paper, we have provided a basic micro-foundation for cross-sectional network models, based on a stochastic choice process.  While not applicable in all circumstances, behaviorally plausible conditions do exist for which this process will give rise to well-defined equilibrium behavior, and for which said equilibrium can be expressed in terms of individual utilities.  The presence of such a relationship allows for estimation and comparison of models for agent preferences from cross-sectional data, without explicit knowledge of the underlying dynamic process.  It is hoped that the results shown here will lead to further investigation of the relationship between choice processes and social networks, and to a clearer understanding of what can (and cannot) be inferred from cross-sectional network data.

%---Uncomment if we need to put the references on a separate page.
%\clearpage 
%\section{References}  
\bibliography{ctb,myref}

\begin{thebibliography}{}

\bibitem[Barndorff-Nielsen, 1978]{barndorffnielsen:bk:1978}
Barndorff-Nielsen, O. (1978).
\newblock {\em Information and Exponential Families in Statistical Theory}.
\newblock John Wiley and Sons, New York.

\bibitem[Besag, 1974]{besag:jrssB:1974}
Besag, J. (1974).
\newblock Spatial interaction and the statistical analysis of lattice systems.
\newblock {\em Journal of the Royal Statistical Society, Series B},
  36(2):192--236.

\bibitem[Brown, 1986]{brown:bk:1986}
Brown, L.~D. (1986).
\newblock {\em Fundamentals of Statistical Exponential Families, with
  Applications in Statistical Decision Theory}.
\newblock Institute of Mathematical Statistics, Hayward, CA.

\bibitem[Butts, 2007]{butts:ch:2007}
Butts, C.~T. (2007).
\newblock Statistical mechanical models for social systems.
\newblock In Bejan, A. and Merkx, G.~W., editors, {\em Constructal Theory of
  Social Dynamics}, pages 197--224. Springer, New York.

\bibitem[Butts, 2008a]{butts:jss:2008a}
Butts, C.~T. (2008a).
\newblock network: a package for managing relational data in {R}.
\newblock {\em Journal of Statistical Software}, 24(2).

\bibitem[Butts, 2008b]{butts:sm:2008}
Butts, C.~T. (2008b).
\newblock A relational event framework for social action.
\newblock {\em Sociological Methodology}, 38(1):155--200.

\bibitem[Butts, 2008c]{butts:jss:2008b}
Butts, C.~T. (2008c).
\newblock Social network analysis with sna.
\newblock {\em Journal of Statistical Software}, 24(6).

\bibitem[Butts, 2011]{butts:sm:2011b}
Butts, C.~T. (2011).
\newblock {B}ernoulli graph bounds for general random graphs.
\newblock {\em Sociological Methodology}, 41:299--345.

\bibitem[Butts, 2015]{butts:jms:2015}
Butts, C.~T. (2015).
\newblock A novel simulation method for binary discrete exponential families,
  with application to social networks.
\newblock {\em Journal of Mathematical Sociology}, 39(3):174--202.

\bibitem[Butts, 2018]{butts:jms:2018}
Butts, C.~T. (2018).
\newblock A perfect sampling method for exponential family random graph models.
\newblock {\em Journal of Mathematical Sociology}, 42(1):17--36.

\bibitem[Butts, 2019]{butts:jms:2019}
Butts, C.~T. (2019).
\newblock A dynamic process interpretation of the sparse {ERGM} reference
  model.
\newblock {\em Journal of Mathematical Sociology}, 43(1):40--57.

\bibitem[Butts, 2021]{butts:jms:2021}
Butts, C.~T. (2021).
\newblock Phase transitions in the edge/concurrent vertex model.
\newblock {\em Journal of Mathematical Sociology}, 43(3):135--147.

\bibitem[Butts, 2022]{butts:jms:2022}
Butts, C.~T. (2022).
\newblock A dynamic process reference model for sparse networks with
  reciprocity.
\newblock {\em Journal of Mathematical Sociology}, 46(1):1--27.

\bibitem[Butts, 2024]{butts:jms:2024}
Butts, C.~T. (2024).
\newblock Continuous time graph processes with known {ERGM} equilibria:
  Contextual review, extensions, and synthesis.
\newblock {\em Journal of Mathematical Sociology}, 48(2):129--171.

\bibitem[Butts and Carley, 2007]{butts.carley:jms:2007}
Butts, C.~T. and Carley, K.~M. (2007).
\newblock Structural change and homeostasis in organizations: A
  decision-theoretic approach.
\newblock {\em Journal of Mathematical Sociology}, 31(4):295--321.

\bibitem[Carley, 1991]{carley:asr:1991}
Carley, K.~M. (1991).
\newblock A theory of group stability.
\newblock {\em American Sociological Review}, 56(3):331--354.

\bibitem[Crouch et~al., 1998]{crouch.et.al:pres:1998}
Crouch, B., Wasserman, S., and Trachtenburg, F. (1998).
\newblock {M}arkov chain {M}onte {C}arlo maximum likelihood estimation for $p*$
  social network models.

\bibitem[Diesser et~al., 2023]{diessner.et.al:jpcB:2023}
Diesser, E.~M., Freites, J.~A., Tobias, D.~J., and Butts, C.~T. (2023).
\newblock Network {H}amiltonian models for unstructured protein aggregates,
  with application to $\gamma${D}-crystallin.
\newblock {\em Journal of Physical Chemistry, B}, 127(3):685--697.

\bibitem[Diessner et~al., 2024]{diessner.et.al:jctc:2024}
Diessner, E.~M., Thomas, L.~J., and Butts, C.~T. (2024).
\newblock Production of distinct fibrillar, oligomeric, and other aggregation
  states from network models of multibody interaction.
\newblock {\em Journal of Chemical Theory and Computation}, 20(18):7829--7840.

\bibitem[Dunbar, 1992]{dunbar:jhe:1992}
Dunbar, R.~I. (1992).
\newblock Neocortex size as a constraint on group size in primates.
\newblock {\em Journal of Human Evolution}, 22(6):469--493.

\bibitem[Fararo, 1978]{fararo:bs:1978}
Fararo, T.~J. (1978).
\newblock An introduction to catastrophes.
\newblock {\em Behavioral Science}, 23:291--317.

\bibitem[Fararo, 1984]{fararo:bs:1984}
Fararo, T.~J. (1984).
\newblock Critique and comment: Catastrophe analysis of the {S}imon-{H}omans
  model.
\newblock {\em Behavioral Science}, 29:212--216.

\bibitem[Frank and Strauss, 1986]{frank.strauss:jasa:1986}
Frank, O. and Strauss, D. (1986).
\newblock {M}arkov graphs.
\newblock {\em Journal of the American Statistical Association}, 81:832--842.

\bibitem[Freidkin, 1998]{friedkin:bk:1998}
Freidkin, N. (1998).
\newblock {\em A Structural Theory of Social Influence}.
\newblock Cambridge University Press, Cambridge.

\bibitem[Gaonkar and Mele, 2021]{gaonkar2021model}
Gaonkar, S. and Mele, A. (2021).
\newblock A model of inter-organizational network formation.
\newblock {\em arXiv preprint arXiv:2105.00458}.

\bibitem[Gilks et~al., 1996]{gilks.et.al:ch:1996}
Gilks, W.~R., Richardson, S., and Spiegelhalter, D.~J. (1996).
\newblock Intoducing {M}arkov chain {M}onte {C}arlo.
\newblock In Gilks, W.~R., Richardson, S., and Spiegelhalter, D.~J., editors,
  {\em Markov Chain Monte Carlo in Practice}, pages 1--20. Chapman and Hall,
  London.

\bibitem[Goodreau et~al., 2008]{goodreau.et.al:jss:2008}
Goodreau, S.~M., Handcock, M.~S., Hunter, D.~R., Butts, C.~T., and Morris, M.
  (2008).
\newblock A statnet tutorial.
\newblock {\em Journal of Statistical Software}, 24(9).

\bibitem[Granovetter, 1973]{granovetter:ajs:1973}
Granovetter, M. (1973).
\newblock The strength of weak ties.
\newblock {\em American Journal of Sociology}, 78(6):1369--1380.

\bibitem[Grazioli et~al., 2019]{grazioli.et.al:jpcB:2019}
Grazioli, G., Yu, Y., Unhelkar, M.~H., Martin, R.~W., and Butts, C.~T. (2019).
\newblock Network-based classification and modeling of amyloid fibrils.
\newblock {\em Journal of Physical Chemistry, B}, 123(26):5452--5462.

\bibitem[H\"{a}ggstr\"{o}m and Jonasson, 1999]{haggstrom.jonasson:jap:1999}
H\"{a}ggstr\"{o}m, O. and Jonasson, J. (1999).
\newblock Phase transition in the random triangle model.
\newblock {\em Journal of Applied Probability}, 36:1101--1115.

\bibitem[Handcock, 2003]{handcock:ch:2003}
Handcock, M.~S. (2003).
\newblock Statistical models for social networks: Inference and degeneracy.
\newblock In Breiger, R., Carley, K.~M., and Pattison, P., editors, {\em
  Dynamic Social Network Modeling and Analysis}, pages 229--240. National
  Academies Press, Washington, DC.

\bibitem[Handcock et~al., 2008]{handcock.et.al:jss:2008}
Handcock, M.~S., Hunter, D.~R., Butts, C.~T., Goodreau, S.~M., and Morris, M.
  (2008).
\newblock {statnet}: Software tools for the representation, visualization,
  analysis and simulation of network data.
\newblock {\em Journal of Statistical Software}, 24(1):1--11.

\bibitem[Holland and Leinhardt, 1981]{holland.leinhardt:jasa:1981}
Holland, P.~W. and Leinhardt, S. (1981).
\newblock An exponential family of probability distributions for directed
  graphs (with discussion).
\newblock {\em Journal of the American Statistical Association},
  76(373):33--50.

\bibitem[Hummon and Fararo, 1995]{hummon.fararo:sn:1995}
Hummon, N.~P. and Fararo, T.~J. (1995).
\newblock Actors and networks as objects.
\newblock {\em Social Networks}, 17:1--26.

\bibitem[Hunter and Handcock, 2006]{hunter.handcock:jcgs:2006}
Hunter, D.~R. and Handcock, M.~S. (2006).
\newblock Inference in curved exponential family models for networks.
\newblock {\em Journal of Computational and Graphical Statistics}, 15:565--583.

\bibitem[Hunter et~al., 2008]{hunter.et.al:jss:2008}
Hunter, D.~R., Handcock, M.~S., Butts, C.~T., Goodreau, S.~M., and Morris, M.
  (2008).
\newblock ergm: A package to fit, simulate and diagnose exponential-family
  models for networks.
\newblock {\em Journal of Statistical Software}, 24(3).

\bibitem[Hunter et~al., 2012]{hunter.et.al:jcgs:2012}
Hunter, D.~R., Krivitsky, P.~N., and Schweinberger, M. (2012).
\newblock Computational statistical methods for social network analysis.
\newblock {\em Journal of Computational and Graphical Statistics}, 21:856--882.

\bibitem[Koskinen and Lomi, 2013]{koskinen.lomi:jsp:2013}
Koskinen, J. and Lomi, A. (2013).
\newblock The local structure of globalization: The network dynamics of foreign
  direct investments in the international electricity industry.
\newblock {\em Journal of Statistical Physics}, 151:523--548.

\bibitem[Krackhardt, 1988]{krackhardt:sn:1988}
Krackhardt, D. (1988).
\newblock Predicting with networks: Nonparametric multiple regression analyses
  of dyadic data.
\newblock {\em Social Networks}, 10:359--382.

\bibitem[Krackhardt, 1992]{krackhardt:ch:1992}
Krackhardt, D. (1992).
\newblock The strength of strong ties: the importance of philos in
  organizations.
\newblock In Nohria, N. and Eccles, R., editors, {\em Networks and
  Organizations: Structures, Form, and Action}, pages 216--239. Harvard
  Business Press, Boston, MA.

\bibitem[Krivitsky et~al., 2011]{krivitsky.et.al:statm:2011}
Krivitsky, P.~N., Handcock, M.~S., and Morris, M. (2011).
\newblock Adjusting for network size and composition effects in
  exponential-family random graph models.
\newblock {\em Statistical Methodology}, 8(4):319--339.

\bibitem[Krivitsky et~al., 2023]{krivitsky.et.al:jss:2023}
Krivitsky, P.~N., Hunter, D.~R., Morris, M., and Klumb, C. (2023).
\newblock {ergm} 4: New features for analyzing exponential-family random graph
  models.
\newblock {\em Journal of Statistical Software}, 105(6):1--44.

\bibitem[Lazega, 1992]{lazega1992analyse}
Lazega, E. (1992).
\newblock Analyse de r{\'e}seaux d'une organisation coll{\'e}giale: les avocats
  d'affaires.
\newblock {\em Revue fran{\c{c}}aise de sociologie}, pages 559--589.

\bibitem[Lee and Butts, 2018]{lee.butts:sn:2018}
Lee, F. and Butts, C.~T. (2018).
\newblock Mutual assent or unilateral nomination? a performance comparison of
  intersection and union rules for integrating self-reports of social
  relationships.

\bibitem[Lerner et~al., 2013]{lerner.et.al:jmp:2013}
Lerner, J., Indlekofer, N., Nick, B., and Brandes, U. (2013).
\newblock Conditional independence in dynamic networks.
\newblock {\em Journal of Mathematical Psychology}, 57(6):275 -- 283.

\bibitem[Lusher et~al., 2012]{lusher.et.al:bk:2012}
Lusher, D., Koskinen, J., and Robins, G. (2012).
\newblock {\em Exponential Random Graph Models for Social Networks: Theory,
  Methods, and Applications}.
\newblock Cambridge University Press, Cambridge.

\bibitem[Macy and Willer, 2002]{macy.willer:ars:2002}
Macy, M.~W. and Willer, R. (2002).
\newblock From factors to actors: Computational sociology and agent-based
  modeling.
\newblock {\em Annual Review of Sociology}, 21:143--166.

\bibitem[Mayhew et~al., 1995]{mayhew.et.al:sf:1995}
Mayhew, B.~H., McPherson, J.~M., Rotolo, M., and Smith-Lovin, L. (1995).
\newblock Sex and race homogeneity in naturally occurring groups.
\newblock {\em Social Forces}, 74(1):15--52.

\bibitem[McFadden, 1973]{mcfadden:ch:1973}
McFadden, D. (1973).
\newblock Conditional logit analysis of qualitative choice behavior.
\newblock In Zarembka, P., editor, {\em Frontiers in Econometrics}. Academic
  Press.

\bibitem[McPherson et~al., 2001]{mcpherson.et.al:ars:2001}
McPherson, J.~M., Smith-Lovin, L., and Cook, J.~M. (2001).
\newblock Birds of a feather: Homophily in social networks.
\newblock {\em Annual Review of Sociology}, 27:415--444.

\bibitem[Mele, 2017]{mele2017structural}
Mele, A. (2017).
\newblock A structural model of dense network formation.
\newblock {\em Econometrica}, 85(3):825--850.

\bibitem[Mele, 2022]{mele2022structural}
Mele, A. (2022).
\newblock A structural model of homophily and clustering in social networks.
\newblock {\em Journal of Business \& Economic Statistics}, 40(3):1377--1389.

\bibitem[Monderer and Shapley, 1996]{monderer.shapley:geb:1996}
Monderer, D. and Shapley, L.~S. (1996).
\newblock Potential games.
\newblock {\em Games and Economic Behavior}, 14:124--143.

\bibitem[Pattison and Robins, 2002]{pattison.robins:sm:2002}
Pattison, P.~E. and Robins, G.~L. (2002).
\newblock Neighborhood-based models for social networks.
\newblock {\em Sociological Methodology}, 32:301--337.

\bibitem[Pattison and Wasserman, 1999]{pattison.wasserman:bjmsp:1999}
Pattison, P.~E. and Wasserman, S. (1999).
\newblock Logit models and logistic regressions for social networks: {II}.
  multivariate relations.
\newblock {\em British Journal of Mathematical and Statistical Psychology},
  52:169--193.

\bibitem[{R Core Team}, 2026]{rteam:sw:2026}
{R Core Team} (2026).
\newblock {\em R: A Language and Environment for Statistical Computing}.
\newblock R Foundation for Statistical Computing, Vienna, Austria.

\bibitem[Radin and Yin, 2013]{radin.yin:aap:2013}
Radin, C. and Yin, M. (2013).
\newblock Phase transitions in exponential random graphs.
\newblock {\em Annals of Applied Probability}, 23(6):2458--2471.

\bibitem[Robins and Pattison, 2001]{robins.pattison:jms:2001}
Robins, G.~L. and Pattison, P.~E. (2001).
\newblock Random graph models for temporal processes in social networks.
\newblock {\em Journal of Mathematical Sociology}, 25:5--41.

\bibitem[Robins et~al., 1999]{robins.et.al:p:1999}
Robins, G.~L., Pattison, P.~E., and Wasserman, S. (1999).
\newblock Logit models and logistic regressions for social networks, {III}.
  valued relations.
\newblock {\em Psychometrika}, 64:371--394.

\bibitem[Robins et~al., 2005]{robins.et.al:ajs:2005}
Robins, G.~L., Pattison, P.~E., and Woolcock, J. (2005).
\newblock Small and other worlds: Network structures from local processes.
\newblock {\em American Journal of Sociology}, 110(4):894--936.

\bibitem[Schweinberger, 2011]{schweinberger:jasa:2011}
Schweinberger, M. (2011).
\newblock Instability, sensitivity, and degeneracy of discrete exponential
  families.
\newblock {\em Journal of the American Statistical Association},
  106:1361--1370.

\bibitem[Schweinberger et~al., 2020]{schweinberger.et.al:ss:2020}
Schweinberger, M., Krivitsky, P.~N., Butts, C.~T., and Stewart, J. (2020).
\newblock Exponential-family models of random graphs: Inference in finite-,
  super-, and infinite-population scenarios.
\newblock {\em Statistical Science}, 35(4):627--662.

\bibitem[Skvoretz et~al., 1996]{skvoretz.et.al:jms:1996}
Skvoretz, J., Faust, K., and Fararo, T. (1996).
\newblock Social structure, networks, and $e$-state structuralism models.
\newblock {\em Journal of Mathematical Sociology}, 21:57--76.

\bibitem[Snijders, 1996]{snijders:jms:1996}
Snijders, T. A.~B. (1996).
\newblock Stochastic actor-oriented models for network change.
\newblock {\em Journal of Mathematical Sociology}, 23:149--172.

\bibitem[Snijders, 2001]{snijders:sm:2001}
Snijders, T. A.~B. (2001).
\newblock The statistical evaluation of social network dynamics.
\newblock {\em Sociological Methodology}, 31:361--395.

\bibitem[Snijders, 2002]{snijders:joss:2002}
Snijders, T. A.~B. (2002).
\newblock {M}arkov chain {M}onte {C}arlo estimation of exponential random graph
  models.
\newblock {\em Journal of Social Structure}, 3(2).

\bibitem[Snijders, 2005]{snijders:ch:2005}
Snijders, T. A.~B. (2005).
\newblock Models for longitudinal network data.
\newblock In Carrington, P.~J., Scott, J., and Wasserman, S., editors, {\em
  Models and Methods in Social Network Analysis}, pages 215--247. Cambridge
  University Press, New York.

\bibitem[Snijders et~al., 2006]{snijders.et.al:sm:2006}
Snijders, T. A.~B., Pattison, P.~E., Robins, G.~L., and Handcock, M.~S. (2006).
\newblock New specifications for exponential random graph models.
\newblock {\em Sociological Methodology}, 36:99{--}154.

\bibitem[Strauss, 1986]{strauss:siam:1986}
Strauss, D. (1986).
\newblock On a general class of models for interaction.
\newblock {\em SIAM Review}, 28(4):513--527.

\bibitem[Strauss and Ikeda, 1990]{strauss.ikeda:jasa:1990}
Strauss, D. and Ikeda, M. (1990).
\newblock Pseudolikelihood estimation for social networks.
\newblock {\em Journal of the American Statistical Association},
  85(409):204--212.

\bibitem[{van Duijn} et~al., 2009]{vanduijn.et.al:sn:2009}
{van Duijn}, M.~A., Gile, K.~J., and Handcock, M.~S. (2009).
\newblock A framework for the comparison of maximum pseudo-likelihood and
  maximum likelihood estimation of exponential family random graph models.
\newblock {\em Social Networks}, 31(1):52--62.

\bibitem[{Vega Yon} et~al., 2021]{von.et.al:sn:2021}
{Vega Yon}, G.~G., Slaughter, A., and {de la Haye}, K. (2021).
\newblock Exponential random graph models for little networks.
\newblock {\em Social Networks}, 64:225--238.

\bibitem[Wasserman and Pattison, 1996]{wasserman.pattison:p:1996}
Wasserman, S. and Pattison, P.~E. (1996).
\newblock Logit models and logistic regressions for social networks: {I}. an
  introduction to {M}arkov graphs and $p*$.
\newblock {\em Psychometrika}, 60:401--426.

\bibitem[Wasserman and Robins, 2005]{wasserman.robins:ch:2005}
Wasserman, S. and Robins, G.~L. (2005).
\newblock An introduction to random graphs, dependence graphs, and $p*$.
\newblock In Carrington, P.~J., Scott, J., and Wasserman, S., editors, {\em
  Models and Methods in Social Network Analysis}, chapter~10, pages 192--214.
  Cambridge University Press, Cambridge.

\bibitem[Yamagishi and Cook, 1993]{yamagishi.cook:spq:1993}
Yamagishi, T. and Cook, K. (1993).
\newblock Generalized exchange and social dilemmas.
\newblock {\em Social Psychology Quarterly}, 56:235--248.

\bibitem[Yu et~al., 2021]{yu.et.al:siam:2021}
Yu, Y., Grazioli, G., Phillips, N.~E., and Butts, C.~T. (2021).
\newblock Local graph stability in exponential family random graph models.
\newblock {\em SIAM Journal on Applied Mathematics}, 81(4):1389--1415.

\end{thebibliography}

%---Uncomment the commands below if we need to display footnotes as endnotes
%\begingroup
%\parindent 0pt
%\parskip 2ex
%\def\enotesize{\normalsize}
%\theendnotes
%\endgroup

\end{document}